\documentclass[11pt]{article}
\usepackage[utf8]{inputenc}
\usepackage{style}
\usepackage{thm-restate}
\usepackage{thmtools}
\usepackage[margin=1in]{geometry}
\usepackage{tikz}
\usepackage{verbatim}
\usepackage{cancel}
\usepackage{titling}
\usepackage{multirow}
\usepackage{quantikz}
\usepackage{bm}
\usepackage{enumitem}
\usepackage{mathtools}
\usepackage{fontawesome5}

\newcommand{\ignore}[1]{}

\definecolor{cb-salmon-pink}{RGB}{255, 182, 119}
\crefname{enumi}{Step}{Steps}
\usepackage{float}
\newcommand{\Tr}{\mathrm{Tr}}

\DeclarePairedDelimiterX{\infdivx}[2]{\lparen}{\rparen}{%
  #1\delimsize\|#2%
}

\title{Nearly Time-Optimal Pure State Tomography\\with Pauli Measurements}
\author{Sabee Grewal\thanks{UT Austin \texttt{sabee@cs.utexas.edu}} \and Meghal Gupta\thanks{UC Berkeley \texttt{meghal@berkeley.edu}} \and William He\thanks{Carnegie Mellon University \texttt{wrhe@cs.cmu.edu}} \and Aniruddha Sen\thanks{UT Austin \texttt{aniruddhasen@utexas.edu}} \and Mihir Singhal\thanks{UC Berkeley \texttt{mihirs@berkeley.edu}}}
\begin{document}
\allowdisplaybreaks
\maketitle
\begin{abstract}
    We give an algorithm for pure state tomography with near-optimal copy and time complexity using only single-qubit measurements. 
    Specifically, given $\widetilde{O}(2^n/\epsilon)$ copies of an unknown $n$-qubit pure state $\ket\psi$, the algorithm performs only \textit{nonadaptive Pauli measurements}, runs in time $\widetilde{O}(2^n/\epsilon)$, and outputs $\ket{\hat{\psi}}$ with fidelity at least $1-\epsilon$ with $\ket{\psi}$ with high probability. 
    This is the first algorithm for pure state tomography that achieves near-optimal running time. 
\end{abstract}

\section{Introduction}

Quantum state tomography is the problem of learning an unknown quantum state from measurement outcomes on independent copies.  In this work, we focus on tomography of \emph{pure} $n$-qubit states.  The learner receives copies of an unknown state $\ket{\psi}\in (\C^2)^{\otimes n}$ and must output an estimate $\ket{\hat\psi}$ with high fidelity, e.g., $1-|\braket{\hat\psi}{\psi}|^2 \le \epsilon$ with probability at least $1-\delta$.\footnote{For pure states, learning with respect to distances such as trace distance, Frobenius distance, and $\chi^2$ distance is equivalent, since each is a monotone function of fidelity.}


If the learner can perform arbitrary measurements, the copy complexity of worst-case pure-state tomography is known to be $\Theta(2^n/\epsilon)$. This rate is achieved by several procedures and is information-theoretically optimal; see, e.g., \cite{hayashi1998asymptotic,haah2016sample,o2016efficient,guctua2020fast,van2023quantum,pelecanos2025mixed,scharnhorst2025optimal,pelecanos2025debiased}. Moreover, it is well known that the measurements do not need to be entangled \emph{across copies} to attain this scaling \cite{voroninski2013quantumtomography,KUENG201788,haah2016sample}. What remains unclear, however, is whether entanglement is required \emph{within each $n$-qubit copy}. 
Such highly entangled measurements are infeasible in practice, and restricting to simpler measurements would make tomography feasible for larger quantum systems.
This motivates the central question of this paper:
\begin{quote}
\emph{Can one achieve the optimal $\Theta(2^n/\epsilon)$ copy complexity using only nonadaptive single-qubit (product basis) measurements?}
\end{quote}
An even more ambitious goal is to achieve optimal pure-state tomography using only Pauli basis measurements (i.e.~measurements diagonalizing operators $\{X,Y,Z\}^{\otimes n}$), a standard and well-studied class of single-qubit measurements. Previously, the strongest guarantee for pure-state tomography using Pauli measurements (and, more generally, any single-qubit measurements) was due to \cite{guctua2020fast}, who achieved copy complexity $\widetilde O(3^n/\epsilon)$.

Their estimator works as follows. First, each copy of the unknown state is measured in a uniformly random Pauli product basis. Each measurement outcome is then converted into a matrix-valued estimate, and these matrices are averaged across samples.\footnote{They then apply a projected least-squares step to enforce positivity and unit trace; if one wishes to output a pure state, one may further post-process by taking the eigenvector corresponding to the largest eigenvalue. Neither of these steps has a substantial effect on the accuracy of the estimate.}  
By construction, this average will equal the true state in expectation. However, the number of samples required for this estimator to concentrate is $\widetilde O(3^n/\epsilon)$.  Moreover, their variance bound is shown to be tight even for pure product states.  
Consequently, any approach that follows this paradigm -- namely, forming an average of (possibly reweighted) matrices derived from Pauli measurement outcomes and arguing concentration via the matrix Bernstein inequality -- cannot asymptotically improve upon the $3^n$ dependence. 


Our main result shows that the $3^n$ scaling is not a fundamental limitation of \emph{all} Pauli measurement schemes.  We give a different (but still simple) learning algorithm that achieves essentially optimal copy complexity while using only Pauli basis measurements.\footnote{The term \emph{Pauli measurement} can refer to two distinct measurement models. In the first, each qubit is measured independently in one of the $X$, $Y$, or $Z$ bases, yielding an $n$-bit outcome; we refer to this as a \emph{Pauli basis measurement}. In the second, the two-outcome projective measurement $\{(I+P)/2,(I-P)/2\}$ associated with an $n$-qubit Pauli operator $P$ is applied, yielding a single outcome in $\{\pm 1\}$; we refer to this as a \emph{Pauli observable measurement}. Under the latter model, it is known that $\Omega(d^2/\eps)$ samples are necessary for pure state tomography \cite{flammia2012compressed, lowe2022lower}.}

\begin{theorem}\label{lem:bound-total-error-intro}
There exists an algorithm that, given copies of an unknown $n$-qubit pure state $\ket{\psi}$, samples $\widetilde O(2^n\log(1/\delta)/\epsilon)$ Pauli product bases, measures one copy of $\ket{\psi}$ in each sampled basis, and outputs an estimate $\ket{\hat{\psi}}$ satisfying $|\langle \hat{\psi}| \psi \rangle|^2 \ge 1 - \epsilon$ with probability at least $1-\delta$.  The algorithm runs in time $\widetilde{O}(2^n/\epsilon)$.
\end{theorem}
In particular, our copy complexity matches (up to polylogarithmic factors) the $\Omega(2^n/\epsilon)$ lower bound that holds even when the learner is allowed \emph{arbitrary} measurements~\cite{scharnhorst2025optimal}.

In addition to achieving near-optimal copy complexity, our algorithm also attains essentially optimal running time. To the best of our knowledge, this is the first algorithm for pure-state tomography with near-optimal running time, addressing a question explicitly raised in~\cite{Grewal2025efficientlearningof}. Before our work, the fastest algorithms required $\widetilde{O}(4^n/\epsilon)$ time~\cite{van2023quantum}, while, among algorithms restricted to Pauli-basis measurements, the best known running time was $\widetilde{O}(8^n/\epsilon)$~\cite{guctua2020fast}. Our results therefore demonstrate that simple, nonadaptive single-qubit measurements suffice to achieve not only near-optimal statistical efficiency, but also near-optimal computational efficiency.

\subsection{Related Work}

\paragraph{Pauli measurement tomography for mixed states.} 
A direct analogue of our work appear in recent works of Acharya, Dharmavarapu, Liu, and Yu~\cite{acharya2025paulinear,acharya2025pauli, yu2020sample} on mixed state tomography using Pauli basis measurements. These works show that (up to polylogarithmic factors) $\Theta(10^n/\eps)$ copies are necessary and sufficient for algorithms using nonadaptive Pauli measurements to perform tomography on general mixed states. 
Interestingly, while nonadaptive Pauli measurements cannot match the copy complexity achieved by general measurement schemes in the mixed-state setting, our work shows that this is not the case for pure state tomography.
Additionally, combining our algorithm with the reduction of \cite{pelecanos2025mixed} yields a sample-optimal (up to polylogarithmic factors) mixed-state tomography algorithm in which all entangling operations are confined to the initial purification step. This is morally similar to settings such as measurement-based quantum computing or magic state distillation, where the more complicated parts of the computation are isolated in an initial preprocessing phase.

A separate line of work studies tomography from \emph{Pauli observable measurements} through the lens of compressed sensing and low-rank matrix recovery, initiated by Gross, Liu, Flammia, Becker, and Eisert~\cite{gross2010compressed} and developed further in, e.g.,~\cite{liu2011universal,flammia2012compressed}.  These works show that a rank-$r$ state in dimension $d$ can be uniquely reconstructed from $\widetilde O(rd)$ randomly chosen Pauli observables when their expectation values are known exactly.  In the tomography setting, however, we only have noisy empirical estimates of these expectation values. Accounting for this statistical noise gives sample complexity bounds on the order of $\widetilde O_\epsilon(r^2 d^2)$ for tomography.
This matches the corresponding lower bound of $\Omega(r^2 d^2/\eps)$ samples for tomography from Pauli observable measurements \cite{flammia2012compressed, lowe2022lower}.



\paragraph{Direct fidelity estimation.} 
We present an algorithm that estimates the Frobenius distance $\norm{\rho - \sigma}_F$ between two (potentially mixed) $n$-qubit quantum states $\rho$ and $\sigma$ using nonadaptive Pauli measurements; this serves as a key subroutine in our pure-state tomography algorithm. 
A closely related problem was studied in the direct fidelity estimation (DFE) procedure of Flammia and Liu~\cite{flammia2011direct}, which applies to the special case where both $\rho$ and $\sigma$ are pure (precisely the regime relevant to our algorithm). 
Their procedure estimates the fidelity between $\rho$ and $\sigma$ to within additive error $\pm \gamma'$, using $\Ot(2^n/\gamma'^2)$ copies. 
Since, for pure states, the fidelity is linearly related to the \emph{squared} Frobenius distance, this yields an estimator for $\norm{\rho-\sigma}_F$ with copy complexity $\Ot(2^n/\gamma'^4)$. 
While this guarantee suffices to obtain a tomography algorithm using $\Ot(2^n/\poly(\eps))$ copies, our stronger Frobenius-distance estimator is required to achieve the optimal dependence on $\eps$.

\paragraph{Quantum state certification.} Related to the task of tomography is the task of \textit{quantum state certification}, in which one is to determine whether an unknown state is close to some hypothesis state given copies of the unknown state. There has been recent work on certifying pure states using single-qubit measurements. See, for example, \cite{huang2025certifying,gupta2025few}. Interestingly, this line of work shows that to avoid exponentially large copy complexities, adaptivity is necessary. This counters our work, which shows that nonadaptive single-qubit measurements are essentially just as powerful as general measurements for the task of pure state tomography.


\section{Technical Overview}
Our algorithm has two components. 
First, we show how to perform pure-state tomography assuming access to an estimator that approximates the Frobenius distance between an unknown state and a candidate state $\sigma$, using only nonadaptive Pauli measurements that do not depend on $\sigma$.
The second component is an implementation of this estimator.

\subsection{Tomography via Frobenius Distance Estimation}
Let $\ket{\psi}$ be a pure $n$-qubit state. Using $\widetilde{O}(2^n/\eps)$ copies, we aim to output $\ket{\hat\psi}$ with
$|\braket{\psi}{\hat\psi}|^2 \ge 1-\eps$.
For any $x\in\{0,1\}^\ell$, let $p_x$ be the probability of obtaining outcome $x$ when measuring the first $\ell$ qubits of $\ket{\psi}$ in the computational basis, and let $\ket{\psi_x}$ be the normalized post-measurement state on the remaining $n-\ell$ qubits conditioned on that outcome.

We reconstruct $\ket{\psi}$ recursively along the binary tree of prefixes. At depth $k$, we maintain estimates $\{\ket{\hat\psi_x} : x\in\{0,1\}^k\}$. Given estimates at depth $k+1$, we will show how to build estimates at depth $k$. The base case is $k=n-1$, where each $\ket{\psi_x}$ is a $1$-qubit state (unique up to global phase). Iterating this process up to $k=0$ will yield $\ket{\hat\psi_\emptyset}\approx \ket{\psi}$.

We now describe how to perform this recursive estimation procedure. Fix $x\in\{0,1\}^k$. Given $\ket{\hat\psi_{x0}}\approx \ket{\psi_{x0}}$ and $\ket{\hat\psi_{x1}}\approx \ket{\psi_{x1}}$, we seek coefficients $\hat\alpha_{x0},\hat\alpha_{x1}$ such that
\[
\ket{\hat\psi_x} \coloneq  \hat\alpha_{x0}\ket{0}\ket{\hat\psi_{x0}} + \hat\alpha_{x1}\ket{1}\ket{\hat\psi_{x1}}
\approx \ket{\psi_x}.
\]
By \Cref{lem:optimal coeffs}, there exist coefficients achieving an error (in Frobenius distance) comparable to the weighted errors of the two children, so it suffices to (approximately) solve this 2-parameter optimization problem. The objective function has a nice closed-form expression as a two-variable trigonometric function, and we essentially need to implement an oracle for evaluating the objective function. We do so with samples from $\ket{\psi_x}$, which we obtain by measuring the first $k$ qubits of $\ket{\psi}$ in the computational basis.


Implementing our low-dimensional optimization problem reduces to estimating the distance between the unknown state $\ket{\psi_x}\bra{\psi_x}$ and the known candidate $\sigma$. This is exactly where the Frobenius-distance estimator from \Cref{sec:frob-overview} is used: it provides, for any fixed $\sigma$, an estimate of $\|\ket{\psi_x}\bra{\psi_x}-\sigma\|_F$ from appropriate measurements on copies of $\ket{\psi_x}$.

The only subtlety is that the measurements on the last $n-k$ qubits must be fixed \emph{before} we learn $x$: we cannot choose the measurement basis adaptively as a function of the observed prefix. To handle this, fix a probability scale $p$ and consider all prefixes with $p_x\approx p$. It suffices to choose a \emph{global} list of $m=\widetilde{O}(2^n p/\eps)$ measurement settings (independent of $x$) that would let the Frobenius estimator compare $\ket{\psi_x}$ to any fixed $\sigma$, which is exactly what \Cref{sec:frob-overview} provides; we then repeat each setting $\Theta((1/p)\poly(n))$ times, always measuring the first $k$ qubits in the computational basis and binning outcomes by the observed $x$. For any $x$ with $p_x\approx p$, with high probability each setting in the global list is applied at least once within the bin for $x$, so we obtain exactly the measurement data needed to run the distance estimator for that $x$. Finally, since the probabilities $p_x$ vary with $x$, we run the above procedure for $p\in\{2^{-n},2^{-(n-1)},\dots,1\}$; every $p_x$ lies within a constant factor of some dyadic $p$, so the appropriate run handles it.


\subsection{Frobenius Distance Estimation}\label{sec:frob-overview}
Suppose that $\rho$ and $\sigma$ are (potentially mixed) $n$-qubit states. We are given access to copies of $\rho$ and to a full classical description of $\sigma$. Our goal is to estimate the quantity $\norm{\rho-\sigma}_F$ up to additive error $\gamma$, using a fixed (nonadaptive) list of measurements that does not depend on either $\rho$ or $\sigma$, and using $O(d/\gamma^2)$ samples.


In the description of our algorithm, we will assume access to samples of both $\rho$ and $\sigma$, rather than a full description of $\sigma$. This is more general, because a classical description of $\sigma$ allows us to simulate measurements on it. 
The first step is to note that $\norm{\rho-\sigma}_F = 2\sqrt{d}\sqrt{\Ex_P[v_P^2]}$, where $v$ is a real vector of length $4^n$, indexed by $n$-qubit Pauli matrices, defined by $v_P=\frac12\Tr(P(\rho-\sigma))$ (so that we always have $|v_P|\leq 1$). 

By measuring $\rho$ and $\sigma$ with respect to the observable $P$, for any fixed $P$ we can draw samples from a Rademacher random variable with mean $v_P$ using just one copy each of $\rho$ and $\sigma$. Therefore, estimating the Frobenius distance using nonadaptive Pauli measurements reduces to the following problem:
estimate $\sqrt{r}\coloneq\sqrt{\Ex_k[v_k^2]}$ to additive error $\alpha$ for vectors $v\in[-1,1]^{N}$, given the ability to specify a multiset $\{k_1,\dots,k_{T}\}$ and receive Rademacher samples with means $v_{k_1},\dots,v_{k_T}$. We need the number of samples $T$ to be $\Ot(1/\al^2)$.

To motivate the algorithm, let us imagine that all entries of $v$ are either \textit{small} (close to $0$) or \textit{big} (have absolute value close to $1$). Write $r=r_{\text{small}}+r_{\text{big}}$, where $r_{\text{small}}$ and $r_{\text{big}}$ are the contributions to $r$ from small and big indices, respectively. 
Our goal is to produce an estimator $\hat{r}$ satisfying $|\hat{r}-r| \lesssim \alpha\sqrt{r} + \alpha^2$, since such an estimate can be converted into an estimator for $\sqrt{r}$ with additive $\alpha$ error. 
At a constant factor loss, it suffices to estimate $\hat{r}_{\text{small}}$ and $\hat{r}_{\text{big}}$ separately with guarantees
\begin{align}
    |\hat{r}_{\text{small}}-r_{\text{small}}| \lesssim \alpha\sqrt{r_{\text{small}}} + \alpha^2 ,\quad |\hat{r}_{\text{big}}-r_{\text{big}}| \lesssim \alpha\sqrt{r_{\text{big}}} + \alpha^2.\label{eq:big small approximations}
\end{align}
\begin{enumerate}
    \item To estimate $r_{\text{small}}$, we sample a small number of indices $k_1,\dots,k_{T_{\text{small}}}$, and take many samples of each $v_{k_t}$ to get accurate estimates of their values, throwing out the values that are too large. We then output the average of the remaining values $v_{k_i}^2$. 
    If $r_{\text{small}}$ is nonnegligible, there must be many small indices, so sampling only $T_{\text{small}}$ indices suffices to hit enough of them.    
    On the other hand, many samples per index are needed to estimate their values accurately enough to satisfy \eqref{eq:big small approximations}. 
    \item To estimate $r_{\text{big}}$, we instead sample a large number $T_\text{big}$ of indices $k_1,\dots,k_{T_{\text{big}}}$, take a small number of samples of each $v_{k_t}$, discard indices whose empirical averages are too small, and output the average of the corresponding values $v_{k_t}^2$.
    Here, many indices are required to ensure that large entries are encountered if they contribute significantly to $r$, but only a few samples per index suffice, since we need only a coarse estimate of $v_{k_t}$ when $v_{k_t}^2$ is large. 
\end{enumerate}
Thresholding based on the empirical values of $v_{k_t}$ introduces two issues. 
First, the resulting estimate of $v_k$ is no longer unbiased when conditioned on an index being classified as small or big.
Second, for indices whose true values lie near the threshold, the thresholding rule may assign the index to neither or both categories with nonzero probability. Consequently, the probabilities with which an index contributes to the small and big parts may not sum to one, leading to under- or over-counting in expectation.
To fix both of these issues at once, before estimating $v_{k_t}$, we run a fixed process that uses the Rademacher samples to classify the index as small or big, so that it only contributes to one of the two categories, and then independently compute the estimate of $v_{k_t}$.

To extend our algorithm to a general algorithm, rather than splitting indices into only two categories, we partition them into $\log(1/\alpha)$ level sets, where the $j$th level consists roughly of indices $k$ with $\abs{v_k} \in [2^{-j},2^{-j+1}]$). 
It turns out that to achieve the desired approximation guarantee, we can take $T_j \approx \al^{-2}/4^j$ indices in level $j$, and take $m_j \approx 4^j$ samples of each index. 
Summing over all levels yields a total sample complexity of $\widetilde{O}(1/\al^2)$ as desired.

\paragraph{Time Efficiency.} Implementing the above algorithm requires one to output Rademacher random variables with expectation $\Tr(\sigma P)$ given a description of $\sigma$ and $P$. Naively, this requires time $2^n$ when $\sigma=\ketbra{\psi}{\psi}$ and we have a classical description of $\ket{\psi}$ (as a list of amplitudes). Since we have to repeat this for roughly $2^n$ different Pauli matrices $P$, this gives a runtime of $4^n$, whereas we desire a runtime that is $\widetilde{O}(2^n)$ (ignoring $\epsilon$ and $\delta$). To perform this Rademacher sampling in near-linear time, we observe the following: Pauli matrices $P$ are phased permutation matrices in the computational basis. If $\ket{\psi}$ has real amplitudes in the computational basis, then $\bra{\psi}P\ket{\psi}$ is simply a weighted sum of values in $\{\pm1,\pm\ii\}$. This observation allows us to use the following algorithm: draw $x\in\{0,1\}^n$ with probability $|\braket{x}{\psi}|^2$ and compute the entry of $P$ that is nonzero in the $x$th column (which can be computed in time $O(n)$ given the description of $P$ as a Pauli string). Some straightforward reweighting and signing gives a value $v$ with which conditioned on $x$ one should output a Rademacher with expectation $v$. The case where $\ket{\psi}$ has phases follows easily.

To draw $x$ with probability $|\braket{x}{\psi}|^2$ we use the alias method (\cite{walker1977efficient}) to construct a data structure in time $\widetilde{O}(2^n)$ at the beginning of the distance estimator that can produce such samples in time $O(1)$.

\section{Preliminaries}

\subsection{Notation}

For mixed states $\rho$ and $\sigma$, let $\norm{\rho - \sigma}_F$
denote the Frobenius distance between the two states. For pure states $\ket{\psi}$ and $\ket{\phi}$, let
\[
    d_F(\ket{\psi}, \ket{\phi}) \coloneq  \norm{\ketbra{\psi}{\psi} - \ketbra{\phi}{\phi}}_F
\]
be shorthand for the Frobenius distance between two pure states.
The Frobenius distance between $\ket{\psi}$ and $\ket{\phi}$ depends only on their overlap, i.e., $d_F(\ket{\psi},\ket{\phi}) = \sqrt{2(1-|\braket{\psi}{\phi}|^2)}$.

\subsection{Concentration Bounds}
We will use the following standard concentration bounds.

\begin{lemma}[Hoeffding for ${[-1,1]}$ random variables] \label{lem:chernoff-11}
Let $Y_1,\dots,Y_n$ be independent random variables with $\E[Y_i]=\mu_i$ and $Y_i\in[-1,1]$ almost surely. Let
\[
S \coloneqq \sum_{i=1}^n (Y_i - \mu_i).
\]
Then for every $\gamma\ge 1$,
\[
\Pr[ |S| > \gamma \sqrt{n} ] \le 2 \exp(-\gamma^2/8).
\]
\end{lemma}

\begin{lemma}[Bernstein for ${[0,B]}$ random variables] \label{lem:chernoff-01}
Let $X_1,\dots,X_m$ be independent random variables with $X_i\in[0,B]$ and let
\[
S \coloneqq \sum_{i=1}^m X_i, \qquad \mu \coloneqq \E[S].
\]
Then for every $\gamma\ge 1$,
\[
\Pr[ \lvert S - \mu\rvert > \gamma(\sqrt{B\mu}+B) ]
  \le 2 \exp(-\gamma/2).
\]
\end{lemma}
\begin{proof}
    If $Y_1,\dots,Y_m$ are
    independent, mean-zero, and satisfy $|Y_i|\le B$, and if
    $V\coloneqq \sum_{i=1}^m \E[Y_i^2]$, then for all $t\ge 0$, then by Bernstein's inequality,
    \begin{align}\label{eq:bernstein}
    \Pr\left[\sum_{i=1}^m Y_i \ge t\right]
    \le \exp\left(-\frac{t^2}{2(V+Bt/3)}\right),
    \end{align}
    and the same holds for the lower tail. 
    Now set $Y_i\coloneq X_i-\E[X_i]$, so $|Y_i|\le B$.  Also
    $\E[Y_i^2]\le \E[X_i^2]\le B\E[X_i]$ since $0\le X_i\le B$, hence
    \[
    V=\sum_i \E[Y_i^2]\le B\sum_i \E[X_i]=B\mu.
    \]
    Use \eqref{eq:bernstein} with $V\le B\mu$ and take $t=\gamma(\sqrt{B\mu}+B)$. We check
    \[
    \frac{t^2}{2(B\mu+Bt/3)} \ge \gamma/2 \qquad (\gamma\ge 1),
    \]
    so $\Pr[S-\mu>t]\le e^{-\gamma/2}$ and similarly $\Pr[\mu-S>t]\le e^{-\gamma/2}$. Using a union bound gives
    \[
    \Pr[|S-\mu|> \gamma(\sqrt{B\mu}+B)]\le 2e^{-\gamma/2}.
    \]
\end{proof}



\section{The Algorithm}

It will be useful to describe $\ket{\psi}$ as a binary tree.
Set
\[
  \ket{\psi_\emptyset} \coloneq  \ket{\psi}, \qquad p_\emptyset \coloneq  1.
\]
For any string $x\in\{0,1\}^{< n}$ (identified by a node in the complete binary tree of depth $n$), we write
\begin{align}
  \ket{\psi_x}
  &= \alpha_{x0}\ket{0}\otimes\ket{\psi_{x0}}
   + \alpha_{x1}\ket{1}\otimes\ket{\psi_{x1}},
\end{align}
where $\ket{\psi_{x0}}$ and $\ket{\psi_{x1}}$ are normalized states on the remaining
$n-\len(x)-1$ qubits and $|\alpha_{x0}|^2 + |\alpha_{x1}|^2 = 1$. Recursively, we define the weight at each node $x$ by
\[
  p_x \coloneq  \Pr[\text{measuring the first $\len(x)$ qubits gives $x$}].
\]
Note that this description is not unique, since phases can be absorbed either into the
$\alpha_{xb}$ or into the conditional states $\ket{\psi_x}$, but any such choice suffices
for our analysis.

The main goal of our algorithm is to reconstruct $\ket{\psi}$ in a ``bottom up'' fashion on this tree: we first learn accurate estimates of the leaf states $\ket{\psi_x}$ for $\len(x) = n-1$, then for all internal nodes with $\len(x) = n-2$, and so on up to the root.

\subsection{Gluing Branches and Error Accumulation}
\subsubsection{Optimal Gluings}
The main subroutine we need is the following: given estimates $\ket{\hat\psi_{x0}}$ and $\ket{\hat\psi_{x1}}$, we wish to glue them together to form an estimate $\ket{\hat\psi_{x}}$.
First, we show that if we have good estimates $\ket{\hat\psi_{x0}}$ and $\ket{\hat\psi_{x1}}$, then there exists a way to glue them together that doesn't worsen the error (but not necessarily that we can find it algorithmically). 

\begin{lemma}\label{lem:optimal coeffs}
Assume estimates $\ket{\hat\psi_{x0}}$ and $\ket{\hat\psi_{x1}}$ for $\ket{\psi_{x0}}$ and $\ket{\psi_{x1}}$ satisfy
\[
    d_F(\ket{\hat\psi_{x0}},\ket{\psi_{x0}}) \le \sqrt{a_0}
    \qquad\text{and}\qquad
    d_F(\ket{\hat\psi_{x1}},\ket{\psi_{x1}})\le \sqrt{a_1}.
\]
Then there exist $\hat\alpha_{x0}$ and $\hat\alpha_{x1}$ with
$|\hat\alpha_{x0}|^2+|\hat\alpha_{x1}|^2=1$ such that
\[
    d_F(\hat\alpha_{x0}\ket0\otimes\ket{\hat\psi_{x0}} +\hat\alpha_{x1}\ket1\otimes\ket{\hat\psi_{x1}},\ket{\psi_x})
    \le \sqrt{\frac{p_{x0}}{p_x}\cdot a_0+\frac{p_{x1}}{p_x}\cdot a_1}.
\]
\end{lemma}

\begin{proof}
Let $w_b\coloneq p_{xb}/p_x$ for $b\in\{0,1\}$ (so $w_0+w_1=1$), and write
\[
    \ket{\psi_x}=\sqrt{w_0}\ket0\otimes\ket{\psi_{x0}}+\sqrt{w_1}\ket1\otimes\ket{\psi_{x1}}
\]
(absorbing any relative phase into $\ket{\psi_{xb}}$).  Set $c_b\coloneq \braket{\psi_{xb}}{\hat\psi_{xb}}$.
For pure states, $d_F(\ket\phi,\ket\psi)^2=1-|\braket{\phi}{\psi}|^2$, hence
\[
    1-|c_b|^2=\frac{1}{2}d_F(\ket{\hat\psi_{xb}},\ket{\psi_{xb}})^2\le \frac{a_b}{2}.
\]
Let $s\coloneq w_0|c_0|^2+w_1|c_1|^2$. If $s=0$ the claim is trivial; otherwise define
\[
    \hat\alpha_{xb}\coloneq \frac{\sqrt{w_b}c_b^*}{\sqrt{s}}\qquad(b\in\{0,1\}),
\]
so $|\hat\alpha_{x0}|^2+|\hat\alpha_{x1}|^2=1$.  Let
\[
    \ket{\hat\psi_x}\coloneq \hat\alpha_{x0}\ket0\otimes\ket{\hat\psi_{x0}}
    +\hat\alpha_{x1}\ket1\otimes\ket{\hat\psi_{x1}}.
\]
Using $\braket{0}{1}=0$,
\[
    \braket{\psi_x}{\hat\psi_x}
    =\sqrt{w_0}\hat\alpha_{x0}c_0+\sqrt{w_1}\hat\alpha_{x1}c_1
    =\sqrt{s},
\]
so
\[
    d_F(\ket{\hat\psi_x},\ket{\psi_x})^2
    =2(1-|\braket{\psi_x}{\hat\psi_x}|^2)
    =2(1-s)
    =2 \sum_{b\in\{0,1\}} w_b(1-|c_b|^2)
    \le w_0 a_0+w_1 a_1.
\]
Taking square roots gives the stated bound.
\end{proof}


\subsubsection{Algorithmic Gluing}
Next, we design an algorithm to approximate the coefficients $\hat\alpha_{x0}, \hat\alpha_{x1}$ guaranteed by \cref{lem:optimal coeffs}.  
This is the subroutine that calls our Frobenius distance estimator. Define for all $x\in\{0,1\}^{\ell}$ the quantity
\begin{align}
    \eps_x &\coloneq \frac{\epsilon}{2^\ell p_x},
\end{align}
where $p_x$ is the probability of obtaining $x$ after measuring the first $\ell$ qubits in the computational basis. 

\begin{lemma}\label{lem:find coeffs}
There exists an adaptive algorithm $\textsc{Find-Coeffs}(\ket{\hat\psi_{x0}},\ket{\hat\psi_{x1}},\ket{\psi_x},\eps^*)$ that satisfies the following provided $\eps^*\le 0.1\eps_x$.
If
\[
d_F(\ket{\hat\psi_{x0}},\ket{\psi_{x0}})
\le (n-\len(x))\sqrt{\eps_{x0}}
\qquad\text{and}\qquad
d_F(\ket{\hat\psi_{x1}},\ket{\psi_{x1}})
\le (n-\len(x))\sqrt{\eps_{x1}},
\]
then with probability at least $1-\delta$, the algorithm outputs
$\hat\alpha_{x0},\hat\alpha_{x1}\in\mathbb{C}$ satisfying $|\hat\alpha_{x0}|^2+|\hat\alpha_{x1}|^2=1$
and
\[
d_F\left(
\hat\alpha_{x0}\ket0\otimes\ket{\hat\psi_{x0}}
+
\hat\alpha_{x1}\ket1\otimes\ket{\hat\psi_{x1}},
\ket{\psi_x}
\right)
\le
(n-\len(x)+1)\sqrt{\eps_x}.
\]
The algorithm consumes $\widetilde{O}({2^{n-\len(x)}\log(1/\delta)}/{\eps^*})$
copies of $\ket{\psi_x}$ and has matching running time.
\end{lemma}
\begin{proof}
Let $m\coloneq n-\len(x)$.
For angles $\theta\in[0,\pi/2]$ and $\gamma\in\mathbb{R}/2\pi\mathbb{Z}$, define
\[
\ket{\chi(\theta,\gamma)}
\coloneq
\cos\theta\ket0\otimes\ket{\hat\psi_{x0}}
+
e^{i\gamma}\sin\theta\ket1\otimes\ket{\hat\psi_{x1}},
\]
and
\[
\Delta(\theta,\gamma)
\coloneq
d_F\bigl(\ket{\chi(\theta,\gamma)},\ket{\psi_x}\bigr).
\]
Since $\ket0\otimes\ket{\hat\psi_{x0}}$ and $\ket1\otimes\ket{\hat\psi_{x1}}$ are orthonormal, any normalized state in their span can be written (up to a global phase) as 
\[
\cos \theta \ket{0}\otimes \ket{\hat\psi_{x0}} + e^{i \gamma} \sin \theta \ket{1} \otimes \ket{\hat\psi_{x1}}
\]
for some $\theta \in [0,\pi/2]$ and $\gamma \in \R/2\pi\Z$. 
Thus it suffices to minimize $\Delta(\theta,\gamma)$ over $(\theta,\gamma)$.

Write
\[
\ket{\psi_x}
=
\alpha_{x0}\ket0\otimes\ket{\psi_{x0}}
+
\alpha_{x1}\ket1\otimes\ket{\psi_{x1}},
\qquad
|\alpha_{x0}|^2+|\alpha_{x1}|^2=1,
\]
and define
\[
c_0\coloneq \alpha_{x0}\braket{\hat\psi_{x0}}{\psi_{x0}},
\qquad
c_1\coloneq \alpha_{x1}\braket{\hat\psi_{x1}}{\psi_{x1}}.
\]
Then
\[
\braket{\chi(\theta,\gamma)}{\psi_x}
=
c_0\cos\theta+c_1e^{-i\gamma}\sin\theta.
\]
For pure states $\ket\psi$ and $\ket{\phi}$, we have $d_F(\ket{\psi},\ket{\phi})^2 = 2\bigl(1-|\braket{\psi}{\varphi}|^2\bigr)$. Thus we obtain
\begin{align}
\frac12\Delta(\theta,\gamma)^2
&=
1-\left|c_0\cos\theta+c_1e^{-i\gamma}\sin\theta\right|^2 \notag\\
&=
1-|c_0|^2\cos^2\theta-|c_1|^2\sin^2\theta
-2\Re\!\left(c_0\overline{c_1}e^{i\gamma}\right)\sin\theta\cos\theta.
\label{eq:delta-squared-basic}
\end{align}
For fixed $\theta$, the last term is minimized by choosing $\gamma^\star\coloneq -\arg(c_0\overline{c_1})$,
so that $\Re\!\left(c_0\overline{c_1}e^{i\gamma^\star}\right)=|c_0||c_1|.$
Define
\[
f(\theta)\coloneq \inf_{\gamma\in\mathbb{R}/2\pi\mathbb{Z}} \Delta(\theta,\gamma).
\]
Then $f(\theta)=\Delta(\theta,\gamma^\star)$, and \eqref{eq:delta-squared-basic} gives
\begin{align}
\frac12 f(\theta)^2
&=
1-|c_0|^2\cos^2\theta-|c_1|^2\sin^2\theta-2|c_0||c_1|\sin\theta\cos\theta \notag\\
&=
1-\bigl(|c_0|\cos\theta+|c_1|\sin\theta\bigr)^2.
\label{eq:f-squared-first}
\end{align}
Let $s\coloneq \sqrt{|c_0|^2+|c_1|^2}$. If $s=0$, then $f(\theta) = \sqrt2$, so every choice of coefficients is valid, and there is nothing to prove.
Assume henceforth that $s>0$, and define $\theta^\star\in[0,\pi/2]$ by
\[
\cos\theta^\star=\frac{|c_0|}{s},
\qquad
\sin\theta^\star=\frac{|c_1|}{s}.
\]
Then $|c_0|\cos\theta+|c_1|\sin\theta=s\cos(\theta-\theta^\star)$, and therefore \eqref{eq:f-squared-first} becomes
\begin{align}
\frac12 f(\theta)^2
&=
1-s^2\cos^2(\theta-\theta^\star).
\label{eq:f-squared-second}
\end{align}
If $\Delta_\star\coloneq \min_{\theta,\gamma}\Delta(\theta,\gamma)$, then \eqref{eq:f-squared-second} implies $\frac12\Delta_\star^2=1-s^2$, and hence
\begin{align}
\frac12 f(\theta)^2
&=
\frac12\Delta_\star^2+s^2\sin^2(\theta-\theta^\star).
\label{eq:f-squared-final}
\end{align}
In particular, $f$ is unimodal on $[0,\pi/2]$, with minimum value $\Delta_\star$ attained at $\theta^\star$.

By \Cref{lem:optimal coeffs} and the assumptions of the lemma, we have
\[
\Delta_\star
\le
\sqrt{\frac{p_{x0}}{p_x}\cdot m^2\eps_{x0}
+\frac{p_{x1}}{p_x}\cdot m^2\eps_{x1}}.
\]
Since
\[
\frac{p_{x0}}{p_x}\eps_{x0}+\frac{p_{x1}}{p_x}\eps_{x1}
=
\frac{p_{x0}}{p_x}\cdot \frac{\epsilon}{2^{\len(x)+1}p_{x0}}
+
\frac{p_{x1}}{p_x}\cdot \frac{\epsilon}{2^{\len(x)+1}p_{x1}}
=
\frac{\epsilon}{2^{\len(x)}p_x}
=
\eps_x,
\]
it follows that $\Delta_\star\le m\sqrt{\eps_x}$.

We now describe the primitive distance oracle.
Fix
\[
\eta\coloneq c\sqrt{\eps^*},
\]
where $c>0$ is a sufficiently small absolute constant chosen later.
For any queried pair $(\theta,\gamma)$, the state $\ket{\chi(\theta,\gamma)}$ has an explicit classical description on $m$ qubits. By \Cref{thm:main-frobenius-time-efficient}, there exists an algorithm that, with probability at least $1-\rho$, consuming 
\[
\widetilde{O}\!\left(\frac{2^m\log(1/\rho)}{\eta^2}\right)
\]
copies of $\ket{\psi_x}$ and time,
outputs an estimate $\widetilde\Delta(\theta,\gamma)$ satisfying
\[
\left|\widetilde\Delta(\theta,\gamma)-\Delta(\theta,\gamma)\right|\le \eta.
\]
Assume now that we have oracle access to additive $\eta$-approximations of $\Delta(\theta,\gamma)$.
We use this oracle access to implement a one-dimensional noisy search primitive. 

\newcommand{\osc}{\mathrm{osc}}
\begin{definition}
    Let $g:I \to \R$ be a function on an interval $I$. We say that $g$ is \emph{$\phi$-regularized} if for all subintervals $[u,u+w]=I'\subseteq I$,
    \begin{align}
        \osc_{I'}(g) &\leq \phi \cdot \osc_{\{u+w/4,u+w/2,u+3w/4\}}(g),
    \end{align}
    where $\osc_{S}(g)=\sup_{a\in S}g(a)-\inf_{b\in S}g(b)$.
\end{definition}

\begin{restatable}{claim}{noisysearch}
\label{claim:noisy-search}
There is an absolute constant $C$ such that the following holds.
Let $g$ be a continuous unimodal, $16$-regularized, and $10$-Lipschitz function on an interval $I=[L,R]$ of width at most $10$.
Assume oracle access to estimates $\widetilde g(t)$ satisfying
\[
|\widetilde g(t)-g(t)|\le \eta
\]
for every queried $t\in I$.
Then there is an adaptive algorithm making
$O(\log(1/\eta))$
queries and returning a point $\hat t\in I$ such that
\[
g(\hat t)\le \inf_{t\in I} g(t)+C\eta.
\]
\end{restatable}

For each fixed $\theta$, define
\[
g_\theta(\gamma)\coloneq \Delta(\theta,\gamma),
\qquad
f(\theta)\coloneq \inf_{\gamma\in[0,2\pi]} g_\theta(\gamma).
\]
Our goal is to find $(\theta,\gamma)$ minimizing $\Delta(\theta,\gamma)$.
We achieve this via a binary-search-like approach, where we first approximately minimize $g_\theta$ over $\gamma$ for each $\theta$, and then minimize the resulting function $f(\theta)$ over $\theta$.
To apply \cref{claim:noisy-search} to $g_\theta$ and $f$, we need the following conditions:

\begin{restatable}{fact}{gregularized}\label{fact:g-regularized}
For every $\theta \in [0,\pi/2]$, the function $g_\theta(\gamma) \coloneq \Delta(\theta,\gamma)$ is unimodal on every interval of length at most $\pi$, $1$-Lipschitz, and $16$-regularized.
\end{restatable}

\begin{restatable}{fact}{fregularized}\label{fact:f-regularized}
The function $f(\theta)\coloneq \inf_{\gamma\in[0,2\pi]} \Delta(\theta,\gamma)$
is unimodal on $[0,\pi/2]$, $2$-Lipschitz, and $16$-regularized.
\end{restatable}

We prove \cref{claim:noisy-search}, \Cref{fact:g-regularized}, and \cref{fact:f-regularized} in \Cref{sec:deferred}.
Fix $\theta\in[0,\pi/2]$.
By \Cref{fact:g-regularized}, the function $g_\theta$ is unimodal on any interval of length at most $\pi$, and is Lipschitz and regularized.
We therefore apply \Cref{claim:noisy-search} separately on the intervals $[0,\pi]$ and $[\pi,2\pi]$ to the function $g_\theta$, and return the better of the two outputs.
This yields a phase $\hat\gamma(\theta)$ such that
\[
\Delta\bigl(\theta,\hat\gamma(\theta)\bigr)
\le f(\theta) + C\eta.
\]
Querying the oracle once more at $(\theta,\hat\gamma(\theta))$, we obtain an estimate $\widetilde f(\theta)$ satisfying
\[
|\widetilde f(\theta)-f(\theta)|\le C'\eta
\]
using $O(\log(1/\eta))$ queries to the oracle for $\Delta$.
Thus we have constructed an oracle for $f(\theta)$ with additive error $O(\eta)$.

By \Cref{fact:f-regularized}, the function $f$ is unimodal on $[0,\pi/2]$, and is Lipschitz and regularized.
We therefore apply \Cref{claim:noisy-search} to the approximate oracle $\widetilde f$.
This yields $\hat\theta\in[0,\pi/2]$ such that
\[
f(\hat\theta)
\le
\inf_{\theta\in[0,\pi/2]} f(\theta)
+
C''\eta
=
\Delta_\star + C''\eta.
\]
Finally, recall that each evaluation of $\widetilde f(\theta)$ produces a phase $\hat\gamma(\theta)$ satisfying
\[
\Delta(\theta,\hat\gamma(\theta)) \le f(\theta) + C\eta.
\]
Applying this at $\theta = \hat\theta$, we obtain a phase $\hat\gamma := \hat\gamma(\hat\theta)$ such that
\[
\Delta(\hat\theta,\hat\gamma)
\le
f(\hat\theta) + C\eta
\le
\Delta_\star + C'''\eta.
\]
Choosing $c>0$ sufficiently small, and recalling that $\eta=c\sqrt{\eps^*}$ and $\eps^*\le 0.1\eps_x$, we may ensure that $C'''\eta \le \sqrt{\eps_x}$.
It follows that
\[
\Delta(\hat\theta,\hat\gamma)
\le
m\sqrt{\eps_x}+\sqrt{\eps_x}
=
(m+1)\sqrt{\eps_x}.
\]
Set
\[
\hat\alpha_{x0}\coloneq \cos\hat\theta,
\qquad
\hat\alpha_{x1}\coloneq e^{i\hat\gamma}\sin\hat\theta.
\]
Then $|\hat\alpha_{x0}|^2+|\hat\alpha_{x1}|^2=1$ and
\[
d_F\left(
\hat\alpha_{x0}\ket0\otimes\ket{\hat\psi_{x0}}
+
\hat\alpha_{x1}\ket1\otimes\ket{\hat\psi_{x1}},
\ket{\psi_x}
\right)
\le
(m+1)\sqrt{\eps_x}.
\]
The total number of queries to the oracle for $\Delta(\theta,\gamma)$ is $O(\log^2(1/\eta))$.
Setting the failure probability of each call to be
\[
\rho = \Theta \left(\frac{\delta}{\log^2(1/\eta)}\right),
\]
a union bound shows that the overall success probability is at least $1-\delta$.
Since each call uses
\[
\widetilde{O}\!\left(\frac{2^m\log(1/\rho)}{\eta^2}\right)
=
\widetilde{O}\!\left(\frac{2^m\log(1/\delta)}{\eps^*}\right)
\]
copies of $\ket{\psi_x}$ and the same runtime, and the factor $\log^2(1/\eta)$ is absorbed into the $\widetilde{O}(\cdot)$ notation, the total sample complexity and runtime are thus
\[
\widetilde{O}\!\left(\frac{2^{n-\len(x)}\log(1/\delta)}{\eps^*}\right).\qedhere
\]
\end{proof}

\subsection{The Overall Algorithm}\label{sec:overall-algorithm}

Let $C$ be a large enough constant depending on the constant in \Cref{lem:find coeffs}. We will first describe the algorithm to determine the Pauli measurements made, and then describe the algorithm to reconstruct the state.

\begin{algorithm}[H]
    \vspace{0.3em}
    \textbf{Input:} Number of qubits $n$, error $\epsilon$.\\
    \textbf{Output:} Measurement multiset $\mathcal{M}$.
    \begin{algorithmic}[1]
        \State $\mathcal{E} \gets \{\epsilon,2\epsilon,4\epsilon,\dots,1\}$; \quad $\mathcal{M} \gets \emptyset$.
        \For{$\ell = n-1,\dots,0$}
            \State $d_\ell \gets 2^{n-\ell}$.
            \For{$\epsilon' \in \mathcal{E}$}
                \State Let $\mathcal{P}_{\ell,\epsilon'}$ be the (nonadaptive) Pauli queries that \textsc{Find-Coeffs} would use on a $d_\ell$-dimensional state at accuracy $\epsilon'$.
                \For{$P \in \mathcal{P}_{\ell,\epsilon'}$}
                    \State Add $C (\eps'/\eps) 2^\ell n^2$ copies of the pair $(\ell,P)$ to $\mathcal{M}$.
                \EndFor
            \EndFor
        \EndFor
        \State \Return $\mathcal{M}$.
    \end{algorithmic}
    \caption{\textsc{Build-Measurement-Set}$(n,\epsilon)$}
    \label{alg:build-measurement-set}
\end{algorithm}

\begin{algorithm}[ht]
    \vspace{0.3em}
    \textbf{Input:} Copies of $\ket{\psi}$, error $\epsilon$, failure probability $\delta$;\\
    \hspace*{4.5em} measurement multiset $\mathcal{M} = \textsc{Build-Measurement-Set}(n,\epsilon)$;\\
    \hspace*{4.5em} outcomes of measuring $\mathcal{M}$ on independent copies of $\ket{\psi}$.\\
    \textbf{Output:} Estimate $\ket{\hat\psi}$.
    \begin{algorithmic}[1]
        \State $\mathcal{E} \gets \{\epsilon,2\epsilon,4\epsilon,\dots,1\}$.
        \For{$\ell = n-1,\dots,0$}
            \State $d_\ell \gets 2^{n-\ell}$.
            \For{$x\in\{0,1\}^\ell$}
                \For{$\epsilon' \in \mathcal{E}$ in increasing order}
                    \State Let $\mathcal{P}_{\ell,\epsilon'}$ be as in \Cref{alg:build-measurement-set}.
                    \If{for every $P\in\mathcal{P}_{\ell,\epsilon'}$ there is at least 1 outcome labeled $(\ell,x,P,\cdot)$} \label{line:if condition}
                            \State Run \textsc{Find-Coeffs}$(\ket{\hat{\psi}_{x0}},\ket{\hat{\psi}_{x1}},\ket{\psi_x},\epsilon')$ using these outcomes (the first outcome if there are multiple) as answers to its Pauli queries, obtaining $\hat{\alpha}_{x0}$ and $\hat{\alpha}_{x1}$.
                            \State \textbf{break}
                    \Else 
                    \State \Return \textsc{Fail}
                    \EndIf
                \EndFor
                \State $\ket{\hat{\psi}_x} \gets \hat{\alpha}_{x0}\ket0\otimes\ket{\hat{\psi}_{x0}} + \hat{\alpha}_{x1}\ket1\otimes\ket{\hat{\psi}_{x1}}$.
            \EndFor
        \EndFor
        \State \Return $\ket{\hat{\psi}_{\emptyset}}$.
    \end{algorithmic}
    \caption{\textsc{Tomography-From-Measurements}$(\ket{\psi},\epsilon,\mathcal{M})$}
    \label{alg:tomography-from-measurements}
\end{algorithm}

\begin{theorem}\label{lem:bound-total-error}
There exists an algorithm $\textsc{Tomography}(\epsilon,\delta)$ that, given copies of an unknown $n$-qubit pure state $\ket{\psi}$, samples $\widetilde O(2^n\poly(n)\log(1/\delta)/\epsilon)$ Pauli product bases, measures one copy of $\ket{\psi}$ in each sampled basis, and outputs an estimate $\ket{\hat{\psi}}$ satisfying $|\langle \hat{\psi}| \psi \rangle|^2 \ge 1 - \epsilon$ with probability at least $1-\delta$.  The algorithm runs in time $\widetilde{O}(2^n\log(1/\delta)/\eps)$.
\end{theorem}

\begin{proof}
The algorithm will simply be to run \Cref{alg:build-measurement-set} and then run \Cref{alg:tomography-from-measurements} with the output of \Cref{alg:build-measurement-set}. The number of measurements outputted by \Cref{alg:build-measurement-set} is the claimed copy complexity of \Cref{lem:bound-total-error}, so we focus on proving its correctness.

\begin{lemma}[Sufficient samples for a node]\label{lem:enough-samples-per-node}
Fix a level $\ell$ and node $x\in\{0,1\}^\ell$ with $\epsilon_x < 2$. Let $d_\ell=2^{n-\ell}$. With probability at least $1-2^{-\Omega(n)}$, the first \textbf{if} condition in \Cref{alg:tomography-from-measurements} succeeds for $x$ at some $\epsilon'\le0.01\epsilon_x$.
\end{lemma}

\begin{proof}
By definition, 
\[
\epsilon_x = 2^{-\ell}\epsilon/p_x \qquad \iff \qquad 
p_x = 2^{-\ell}\epsilon/\epsilon_x.
\]
Thus when measuring the first $\ell$ qubits of $\ket{\psi}$ in the computational basis, the probability of obtaining prefix $x$ is $p_x \geq 2^{-\ell}\epsilon/\epsilon_x$.

Let $\epsilon'$ be the largest value in $\{\epsilon,2\epsilon,4\epsilon,\dots,1\}$ satisfying $\epsilon'\le0.01\epsilon_x$. Then since $\epsilon_x<2$ it must be that $\epsilon'/\epsilon_x\ge 0.005$. For this $\epsilon'$ and each $P\in\mathcal{P}_{\ell,\epsilon'}$, the measurement multiset contains
\[
N_{\ell,\epsilon'} \coloneq C(\epsilon'/\epsilon)2^\ell n^2
\]
copies of $(\ell,P)$ (see \cref{alg:build-measurement-set}). Across these $N_{\ell,\epsilon'}$ trials, the number of times we obtain prefix $x$ is $\mathrm{Binomial}(N_{\ell,\epsilon'},p_x)$ with expectation
\[
\Ex[\#\text{hits of }x]
= N_{\ell,\epsilon'} p_x
\ge
C(\epsilon'/\epsilon)2^\ell n^2 \cdot 2^{-\ell}\epsilon/\epsilon_x
= C(\epsilon'/\epsilon_x)n^2
\ge Cn^{2}.
\]
Therefore,
\[
\Pr[\text{no outcome labeled }(\ell,x,P,\cdot)]
= (1-p_x)^{N_{\ell,\epsilon'}}
\le \exp(-p_x N_{\ell,\epsilon'})
\le \exp(-\Omega(n^{2})).
\]
A union bound over all $P\in\mathcal{P}_{\ell,\epsilon'}$ shows that with high probability every $P\in\mathcal{P}_{\ell,\epsilon'}$ has at least one recorded outcome labeled $(\ell,x,P,\cdot)$, so the \textbf{if} condition in \Cref{alg:tomography-from-measurements} holds for this $\epsilon'\le0.01\epsilon_x$.
\end{proof}
\begin{lemma}[Node distance invariant]\label{lem:node-fidelity-invariant}
Under the high-probability event of \Cref{lem:enough-samples-per-node}, for every level $\ell$ and node $x\in\{0,1\}^\ell$ with $\epsilon_x=2^{-\ell} \eps / p_x$, the estimate  $\ket{\hat\psi_x}$ satisfies
\[
d_F(\ket{\psi_x},\ket{\hat\psi_x})\leq (n-\ell+1)\sqrt{\epsilon_x},
\]
where $\ket{\psi_x}$ is the true normalized conditional state.
\end{lemma}

\begin{proof}
    We will prove this statement by induction. The base case of $\ell=n$ is trivial, because each $\ket{\psi_x}$ is just the 0 qubit state. The high probability event of \Cref{lem:enough-samples-per-node} guarantees that the conclusion of \Cref{lem:find coeffs} applies, completing the proof.
\end{proof}

\Cref{lem:node-fidelity-invariant} finishes the proof of \Cref{lem:bound-total-error}. 
In particular for the root $x=\emptyset$ we have $d_F(\ket{\psi},\ket{\hat\psi}) \leq n\sqrt{\eps}$. Therefore, we have

\begin{align*}
&~ d_F(\ket{\psi},\ket{\hat\psi}) = 2\bigl(1-|\braket{\psi}{\hat\psi}|^2\bigr) \\
\implies &~ |\braket{\psi}{\hat\psi}|^2 \ge 1-\frac{n^2\eps}{2}, 
\end{align*}
and using $2\e/n^2$ in place of $\e$ gives the desired result. 
\Cref{alg:build-measurement-set} runs in the desired runtime. The runtime guarantee of \Cref{alg:tomography-from-measurements} follows because for each $x\in\{0,1\}^{\leq n}$ the algorithm uses time $O(2^{n-\mathrm{len}(x)})$ and one call of \textsc{Find-Coeffs} with $\eps^*=\Omega(\eps)$ which takes time $2^{n-\mathrm{len}(x)}\cdot \poly(1/\eps)\cdot\log(1/\delta)$ by \Cref{lem:find coeffs}. Summing over all $x$ gives the result.
\end{proof}

\section{Frobenius Distance Estimation}\label{sec:frob-est}


We present our algorithm for estimating the Frobenius distance between two (potentially mixed) quantum states using nonadaptive Pauli observable measurements. 

\begin{definition}[Nonadaptive Pauli measurement scheme]\label{def:nonadaptive}
A \textit{nonadaptive Pauli measurement scheme} with $M$ measurements is a (possibly randomized) procedure that chooses Pauli observables
\[
P_1,\dots,P_M \in \{I,X,Y,Z\}^{\otimes n}
\]
before any measurement is performed. Given an $n$-qubit quantum state $\rho$, for each $i \in \{1,\dots,M\}$, the scheme measures $P_i$ on an independent copy of $\rho$ and records an outcome $X_i \in \{-1,+1\}$ satisfying
\[
\Ex[X_i] = \operatorname{Tr}(\rho P_i).
\]
\end{definition}

\begin{theorem}[Frobenius distance estimation]\label{thm:main-frobenius}
Let $\rho$ and $\sigma$ be quantum states on $n$ qubits, and let $d = 2^n$. There exists a nonadaptive Pauli measurement scheme (depending only on $d$, $\gamma$, and $\delta$) using $\widetilde{O}(d\log(1/\delta)/\gamma^2)$ measurements on independent copies of $\rho$ and $\sigma$ that outputs an estimate $\hat{D}\in \mathbb{R}$ satisfying
\[
\left| \hat{D} - \left\| \rho - \sigma \right\|_F \right| \le \gamma
\]
with failure probability at most $(\gamma/d)^{10}\delta$.
\end{theorem}

The main ingredient in our proof will be the following theorem about learning classical distributions.

\begin{theorem}[Classical Rademacher norm estimation]\label{thm:rademacher-norm}
Let $v = (v_1,\dots,v_N) \in [-1,1]^N$ be an unknown vector.  
Suppose that we may (nonadaptively) pick queries $k_1,\dots,k_M \in \{1,\dots,N\}$ (possibly with repetitions). 
For each $j \in \{1,\dots,M\}$, we then receive a sample of a Rademacher random variable $X_j \in \{-1,+1\}$ with $\Ex[X_j] = v_{k_j}$. Then, it is possible to make $M = O(1/\alpha^2)$ nonadaptive queries and output an estimate $\qhat$ such that
\[
\left| \qhat - \sqrt{\Ex_{k \ot [N]}[v_k^2]} \right| \le \alpha,
\]
with probability at least $2/3$.
\end{theorem}

We will start by reducing our main theorem to \Cref{thm:rademacher-norm}.

\begin{proof}[Reduction of \texorpdfstring{\cref{thm:main-frobenius}}{Theorem \ref{thm:main-frobenius}} to \texorpdfstring{\Cref{thm:rademacher-norm}}{Theorem \ref{thm:rademacher-norm}}]
Write $\delta = \rho - \sigma$.  
Expanding in the Pauli basis, $\delta = d^{-1} \sum_{P} \delta_P P$ with $\delta_P = \operatorname{Tr}(\delta P)$, and orthogonality of the Paulis implies
\[
\|\delta\|_F^2 = d^{-1} \sum_{P} \delta_P^2.
\]
Define $v_P = \tfrac{1}{2}\delta_P$.  
Since $\left|\operatorname{Tr}(\rho P)\right| \le 1$ and likewise for $\sigma$, we have $|v_P| \le 1$.  
If we let the expectation $\Ex_P[\cdot]$ be over a uniformly random Pauli label $P$, then
\[
\|\rho - \sigma\|_F^2
= d^{-1} \sum_{P} \delta_P^2
= d^{-1} \sum_{P} (2 v_P)^2
= 4d \Ex_P[v_P^2],
\]
so
\[
\|\rho - \sigma\|_F = 2\sqrt{d}\sqrt{\Ex_P[v_P^2]}.
\]
Note that, for each $P$, we may obtain a Rademacher ($\pm 1$) random variable with expectation $v_P$ using one sample each of $\rho$ and $\sigma$, as follows. Measure $P$ on $\rho, \sigma$ (respectively) to obtain $X, Y \in \{\pm 1\}$. Then return the random variable $Z$, which is $X$ or $-Y$ with probability $1/2$ each.
Note that $\Ex[Z] = \tfrac{1}{2}(\operatorname{Tr}(\rho P) - \operatorname{Tr}(\sigma P)) = v_P$, as desired. Thus, we may indeed obtain Rademacher queries as required in \cref{thm:rademacher-norm}, using one (nonadaptive) Pauli measurement per query.

Thus, by \Cref{thm:rademacher-norm}, we may make $M = O(1/\alpha^2)$ nonadaptive measurements to obtain $\qhat$ such that
\[
\left| \qhat - \sqrt{\Ex_P[v_P^2]} \right| \le \alpha
\]
with probability at least $2/3$. Define $\hat{D} = 2\sqrt{d}\qhat$. Then
\[
\left| \hat{D} - \|\rho - \sigma\|_F \right|
= 2\sqrt{d}\left| \qhat - \sqrt{\Ex_P[v_P^2]} \right|
\le 2\sqrt{d}\alpha.
\]
Choosing $\alpha = \gamma/(2\sqrt{d})$ yields $|\hat{D} - \|\rho - \sigma\|_2| \le \gamma$ with probability at least $2/3$.  

Each classical query uses $O(1)$ Pauli measurements on independent copies of $\rho$ and $\sigma$, so the total number of Pauli measurements is
\[
M = O\bigl(1/\alpha^2\bigr) = O\bigl(d/\gamma^2\bigr).
\]
Standard repetition and taking the median amplifies the success probability to at least $1 - (\gamma/d)^{10}\delta$ at the cost of only polylogarithmic factors in $d$ and $1/\gamma$ and $1/\delta$.
\end{proof}

\subsection{Time Efficiency}
For our tomography algorithm, we only use \Cref{thm:main-frobenius} when $\sigma=\ketbra{\psi}{\psi}$ is a known pure state. 
In this case, we can improve the time efficiency of our Frobenius distance estimator:
\begin{theorem}\label{thm:main-frobenius-time-efficient}
Set $d=2^n$. Let $\rho$ be a mixed quantum state on $n$ qubits. Let $\ket{\psi}$ be a pure quantum state on $n$ qubits. There exists an algorithm that takes as input a description of $\ket{\psi}$ as a list of $2^n$ amplitudes in the computational basis and the output of a nonadaptive Pauli measurement scheme (depending only on $d$, $\gamma$, and $\delta$) using $\widetilde{O}(d\log(1/\delta)/\gamma^2)$ measurements on independent copies of $\rho$ and outputs an estimate $\hat{D}\in \mathbb{R}$ satisfying
\[
\left| \hat{D} - \left\| \rho - \ketbra{\psi}{\psi} \right\|_F \right| \le \gamma
\]
with failure probability at most $(\gamma/d)^{10}\delta$. The algorithm runs in time $\widetilde{O}(d\log(1/\delta)/\gamma^2)$.
\end{theorem}
\begin{proof}
    \Cref{alg:choose-indices} and \Cref{alg:build-estimator} both run in the desired runtime, so it suffices to execute the reduction of \Cref{thm:main-frobenius} to \Cref{thm:rademacher-norm} in the desired runtime. To execute the reduction we must be able to sample for any Pauli a collection of $\widetilde{O}(d\log(1/\delta)/\gamma^2)$ Rademacher random variable $Z$ as in the proof of the reduction in the given runtime. This consists of sampling the corresponding $X$ and $Y$. Sampling $X$ can be done in the time it takes to measure $\rho$ in a Pauli basis, which is $\mathrm{poly}(n)=\widetilde{O}(1)$. \Cref{prop:pauli-sampler} shows how to sample the Rademacher $Y$ with $\Ex[Y] = \braket{\psi}{P|\psi}$ as needed in the needed runtime.
\end{proof}

\begin{proposition}\label{prop:pauli-sampler}
    Let $\ket{\psi}$ be a pure quantum state. There exists a data structure and algorithm that uses space $O(d)$ that given a query in the form of a Pauli matrix $P$, runs in time $O(\log d)$, and outputs a Rademacher variable with expectation $\bra{\psi}P\ket{\psi}$. Given a description of $\ket{\psi}$ as a list of $2^n$ amplitudes in the computational basis, it takes time $\widetilde{O}(d)$ to set up.
\end{proposition}
\begin{proof}
    Using \Cref{thm:alias-method} we build a data structure that takes the desired space, takes time $\widetilde{O}(d)$ to set up, and takes time $O(1)$ to answer queries with an independent sample from the distribution $\calD_{\ket{\psi}}$ on $\{0,1\}^n$ given by $\Pr[x]=|\braket{x}{\psi}|^2$. The data structure is based on the alias method:
    \begin{theorem}[\cite{walker1977efficient}]\label{thm:alias-method}
        Let $p$ be a probability distribution on $d$ elements. There exists a data structure and algorithm that takes $O(d)$ space and time $O(1)$ to answer a query with an independent sample from $p$. Setting up the data structure takes time $\widetilde{O}(d)$ given access to the list of probabilities of $p$. 
    \end{theorem}
    We now describe how to turn an independent sample from $\calD_{\ket{\psi}}$ and a Pauli $P$ into a Rademacher random variable with expectation $\bra{\psi}P\ket{\psi}$. First write
    \[
    P = \omega \bigotimes_{j=1}^n P_j,
    \]
    where each \(P_j\in\{I,X,Y,Z\}\), and \(\omega\in\{\pm 1,\pm i\}\).
    
    Every Pauli string acts on computational basis vectors as a signed permutation. Concretely, there exists a bit string \(a\in\{0,1\}^n\) and a phase function \(\eta:\{0,1\}^n\to\{\pm 1,\pm i\}\) such that for every \(x\in\{0,1\}^n\),
    \[
    P\ket{x} = \eta(x)\ket{x\oplus a}.
    \]
    Here \(a_j=1\) exactly when \(P_j\in\{X,Y\}\), and \(a_j=0\) exactly when \(P_j\in\{I,Z\}\). Thus
    \[
    \bra{\psi}P\ket{\psi}
    =
    \sum_{x\in\{0,1\}^n}\overline{\psi_x}\eta(x)\psi_{x\oplus a}.
    \]
    We split into two cases:
    
    \noindent\textbf{Case 1: \(a=0\).} In this case \(P\) is diagonal in the computational basis, so \(\eta(x)\in\{\pm 1\}\) for all \(x\), and
    \[
    P\ket{x}=\eta(x)\ket{x}.
    \]
    The algorithm is simply:
    \begin{enumerate}
        \item Draw \(x\) from the oracle distribution \(\Pr[x]=|\psi_x|^2\).
        \item Output \(\bm{b}:=\eta(x)\in\{\pm 1\}\).
    \end{enumerate}
    Then
    \[
    \Ex[\bm{b}]
    =
    \sum_x |\psi_x|^2 \eta(x)
    =
    \sum_x \overline{\psi_x}\eta(x)\psi_x
    =
    \bra{\psi}P\ket{\psi}.
    \] 
    \noindent\textbf{Case 2: \(a\neq 0\).}
    Now the map \(x\mapsto x\oplus a\) has no fixed points, so \(\{0,1\}^n\) is partitioned into disjoint pairs
    \[
    \{x,x\oplus a\}.
    \]
    Choose a canonical representative \(r\) for each pair, for example the lexicographically smaller of the two strings. Then
    \[
    \bra{\psi}P\ket{\psi}
    =
    \sum_{r}
    \left(
    \overline{\psi_r}\eta(r)\psi_{r\oplus a}
    +
    \overline{\psi_{r\oplus a}}\eta(r\oplus a)\psi_r
    \right),
    \]
    where the sum runs over all representatives \(r\).
    For each representative \(r\), define
    \[
    m_r := |\psi_r|^2 + |\psi_{r\oplus a}|^2
    \]
    and
    \[
    t_r :=
    \overline{\psi_r}\eta(r)\psi_{r\oplus a}
    +
    \overline{\psi_{r\oplus a}}\eta(r\oplus a)\psi_r.
    \]
    Since \(P\) is Hermitian, the quantity \(t_r\) is real. Also,
    \[
    |t_r|
    \leq
    2|\psi_r||\psi_{r\oplus a}|
    \leq
    |\psi_r|^2+|\psi_{r\oplus a}|^2
    =
    m_r,
    \]
    and hence whenever \(m_r>0\),
    \[
    y_r:=\frac{t_r}{m_r}\in[-1,1].
    \]
    We now describe the algorithm.
    \begin{enumerate}
        \item Draw \(x\) from $\calD_{\ket{\psi}}$.
        \item Compute \(x\oplus a\), and let
        \[
        r := \min_{\mathrm{lex}}\{x,x\oplus a\}.
        \]
        \item Look up the two amplitudes \(\psi_r\) and \(\psi_{r\oplus a}\).
        \item Compute $ m_r = |\psi_r|^2+|\psi_{r\oplus a}|^2$ and $t_r=
        \overline{\psi_r}\eta(r)\psi_{r\oplus a}
        +
        \overline{\psi_{r\oplus a}}\eta(r\oplus a)\psi_r$.
        If \(m_r=0\), then this pair is never sampled, so this case may be assigned arbitrarily. Otherwise set
        \[
        y:=\frac{t_r}{m_r}\in[-1,1].
        \]
        \item Output \(\bm{b}\in\{\pm 1\}\) according to
        \[
        \Pr[\bm{b}=1\mid r]=\frac{1+y}{2},
        \qquad
        \Pr[\bm{b}=-1\mid r]=\frac{1-y}{2}.
        \]
    \end{enumerate}
    Conditioned on the event that the sampled pair is \(r\), we therefore have
    \[
    \Ex[\bm{b}\mid r]=y_r=\frac{t_r}{m_r}.
    \]
    Now the probability that the oracle sample lands in the pair indexed by \(r\) is exactly
    \[
    m_r=|\psi_r|^2+|\psi_{r\oplus a}|^2,
    \]
    since this happens iff the sample is either \(r\) or \(r\oplus a\). Therefore
    \begin{align*}
    \Ex[\bm{b}]
    &=
    \sum_r \Pr[\text{sampled pair }=r]\;\Ex[\bm{b}\mid r] \\
    &=
    \sum_r m_r \cdot \frac{t_r}{m_r} =
    \sum_r t_r =
    \bra{\psi}P\ket{\psi}.
    \end{align*}
    This proves correctness. For runtime, once a single oracle sample \(x\) has been drawn, the remaining work consists of computing \(x\oplus a\), choosing the canonical representative \(r\), evaluating the phase function \(\eta\), reading two amplitudes from the explicit description of \(\ket{\psi}\), and performing \(O(1)\) arithmetic operations on these values. All of this takes \(O(\log d)\) time.
\end{proof}

\subsection{Proof of \texorpdfstring{\Cref{thm:rademacher-norm}}{Theorem \ref{thm:rademacher-norm}}} \label{sec:rademacher-proof}

We start by describing the sampling procedure and the estimator built from the resulting outcomes.
\begin{algorithm}[H]
    \vspace{0.3em}
    \textbf{Input:} Accuracy parameter $\alpha \in (0,1)$; query access to $v \in [-1,1]^N$.\\
    \textbf{Output:} Indices and samples $\{(k_{j,t}, X_{j,t,a})\}$.
    \begin{algorithmic}[1]
        \State Set $J \gets  \log_2(1/\alpha) $, $m_0 \gets  2000 \log(1/\alpha) $.
        \For{$j = 0,\dots,J$}
            \State $T_j \gets  \alpha^{-2}/4^j $, \quad $m_j \gets 4^j m_0$.
            \For{$t = 1,\dots,T_j$}
                \State Sample $k_{j,t} \in [N]$ uniformly.
                \For{$a = 1,\dots,m_j$}
                    \State Query $k_{j,t}$ once to obtain $X_{j,t,a} \in \{-1,+1\}$.
                \EndFor
            \EndFor
        \EndFor
        \State \Return $\{(k_{j,t}, X_{j,t,a}) : j \in \{0,\dots,J\}, t \in [T_j], a \in [m_j]\}$.
    \end{algorithmic}
    \caption{\textsc{Choose-Indices}$(\alpha)$}
    \label{alg:choose-indices}
\end{algorithm}

\vspace{-.16in}
\begin{algorithm}[H]
    \vspace{0.3em}
    \textbf{Input:} Accuracy parameter $\alpha$; samples $\{(k_{j,t}, X_{j,t,a})\}$ from \textsc{Choose-Indices}$(\alpha)$.\\
    \textbf{Output:} Estimate $\qhat$ of $\sqrt{\Ex_{k}[v_k^2]}$.
    \begin{algorithmic}[1]
        \State Set $J \gets  \log_2(1/\alpha) $, $m_0 \gets 1000 \log(1/\alpha) $, $n_0 \gets m_0 / 4$.

        \Statex
        \Function{Est-v-Squared}{$j,t$}
            \State $m_j \gets 4^j m_0$.
            \State $\mu^{(1)} \gets \frac{4}{m_j} \sum_{a=1}^{m_j/4} X_{j,t,a}$.
            \State $\mu^{(2)} \gets \frac{4}{m_j} \sum_{a=m_j/4+1}^{m_j/2} X_{j,t,a}$.
            \State \Return $\mu^{(1)} \mu^{(2)}$.
        \EndFunction

        \Statex
        \Function{Level-Check}{$j,t$}
            \State $m_j \gets 4^j m_0$.
            \For{$b = 0,\dots,j$}
                \State $n \gets 4^b n_0$.
                \State $\overline{X} \gets \frac{1}{n} \sum_{a=m_j/2+1}^{m_j/2+n} X_{j,t,a}$.
                \If{$|\overline{X}| > 2^{-b}$}
                    \State \Return $\mathbf{1}[b = j]$.
                \EndIf
            \EndFor
            \State \Return 0.
        \EndFunction

        \Statex
        \For{$j = 0,\dots,J$}
            \State $T_j \gets  \alpha^{-2}/4^j $.
            \For{$t = 1,\dots,T_j$}
                \State $U_{j,t} \gets \Call{Est-v-Squared}{j,t}$.
                \State $\widetilde U_{j, t} \gets \min(U_{j, t}, 16 \cdot 4^{-j})$
                \State $Z_{j,t} \gets \Call{Level-Check}{j,t}$.
                \State $\rhat_{j,t} \gets \widetilde U_{j,t} Z_{j,t}$.
            \EndFor
            \State $\rhat_j \gets \frac{1}{T_j} \sum_{t=1}^{T_j} \rhat_{j,t}$.
        \EndFor
        \State \Return $\qhat = \sqrt{\sum_{j=0}^J \rhat_j}$.
    \end{algorithmic}
    \caption{\textsc{Build-Estimator}$(\alpha,\{(k_{j,t},X_{j,t,a})\})$}
    \label{alg:build-estimator}
\end{algorithm}

Fix $x \in [-1,1]$.  To define $L(x)$, consider the following infinite version of \textsc{Level-Check} that has direct access to an i.i.d.\ stream of Rademacher random variables of mean $x$. Draw $n_b = 4^J n_0$ samples $Y^{(b)}_1,\dots,Y^{(b)}_{n_b}$ with $\E[Y^{(b)}_i] = x$, and let
\[
\overline{ Y_b} \coloneqq \frac{1}{n_b} \sum_{i=1}^{n_b} Y^{(b)}_i.
\]
If there exists $b \le J$ with $|\overline{Y_b}| > 2^{-b}$, let $L(x)$ be the smallest such $b$; otherwise set $L(x) = J+1$.

By construction, $L(x)$ is the first level at which the empirical mean looks ``large'' relative to the threshold $2^{-b}$.  We couple the randomness in $\textsc{Level-Check}(j,t)$ with the randomness in $L(\cdot)$ so that, conditioned on $k_{j,t}$, the random variable
\[
Z_{j,t} = \textsc{Level-Check}(j,t)
\]
has the same distribution as $\ind{L(v_{k_{j,t}}) = j}$.  In particular, we write
\[
z_{j,t} \coloneqq \E[Z_{j,t} \mid k_{j,t}] = \Pr[L(v_{k_{j,t}}) = j \mid k_{j,t}].
\]
We now prove several concentration lemmas relating $v_k$, $L(v_k)$, and the random variables $U_{j,t}$ and $Z_{j,t}$.

\begin{lemma} \label{lem:level-concentration}
For any fixed $x \in [-1,1]$, with probability $1-O(\alpha^2)$ over the randomness defining $L(x)$, the following holds:
\begin{itemize}
    \item If $L(x) = j \in \{0,\dots,J\}$ then
    \[
        0.9 \cdot 2^{-j} \le \len(x) \le 2.2 \cdot 2^{-j}.
    \]
    \item If $L(x) = J+1$ then
    \[
        \len(x) \le 2.2 \cdot 2^{-J} = O(\alpha).
    \]
\end{itemize}
\end{lemma}
\begin{proof}
For each $b = 0,\dots,J$, let $\overline{X}_b$ be the empirical mean of the $4^b n_0$ Rademacher samples (of mean $x$) used at level $b$. By the Hoeffding bound (\Cref{lem:chernoff-11}), we have
\[
\Pr[|\overline{X}_b - x| > 0.1 \cdot 2^{-b}]
    \le O(\alpha^{10}),
\]
using $n_0 = 1000 \log(1/\alpha^2)$.  
By a union bound over $b = 0,\dots,J = O(\log(1/\alpha^2))$, with probability $1-O(\alpha^2)$ we have
\begin{equation} \label{eq:Xb-conc}
|\overline{X}_b - x| \le 0.1 \cdot 2^{-b} \qquad\text{for all } b.
\end{equation}
Assuming \eqref{eq:Xb-conc} holds for every $b$, the conclusion immediately follows by the definition of $L(x)$.
\end{proof}

\begin{lemma} \label{lem:U-concentration}
For each $j \le J$ and $t$, conditioned on $k_{j,t}$, with probability $1-O(\alpha^2)$ we have
\[
|\mu^{(1)} - v_{k_{j,t}}| \le 0.05 \cdot 2^{-j}
\quad\text{and}\quad
|\mu^{(2)} - v_{k_{j,t}}| \le 0.05 \cdot 2^{-j},
\]
where $\mu^{(1)},\mu^{(2)}$ are the quantities computed in \textsc{Est-v-Squared}.
\end{lemma}
\begin{proof}
Condition on $k_{j,t}$, each $\mu^{(\ell)}$ is the average of $m_j/4 = 4^{j-1} m_0$ i.i.d.\ Rademacher variables with mean $v_{k_{j,t}}$ and range in $[-1,1]$. Thus, the conclusion again follows directly from a Hoeffding bound (\cref{lem:chernoff-11}).
\end{proof}

\begin{lemma} \label{lem:truncation}
For each $j \le J$ and $t$, conditioned on $k_{j,t}$, with probability $1-O(\alpha^2)$ we have
\[
U_{j,t} Z_{j,t} = \Ut_{j,t} Z_{j,t}.
\]
\end{lemma}
\begin{proof}
Fix $j,t$ and condition on $k_{j,t}$, writing $x \coloneq  v_{k_{j,t}}$. If $Z_{j,t} = 0$ then both sides are $0$, so we will study what happens when $Z_{j,t}=1$.  

Recall that by how we defined $L(x)$ (running the same \textsc{Level-Check} procedure to all levels), we can couple so that $Z_{j,t}=1$ if and only if $L(x)=j$, while having failure probability $O(\alpha^2)$.
Therefore, by \cref{lem:level-concentration} applied to $x$, with probability $1-O(\alpha^2)$ we have
\[
Z_{j,t}=1 \implies \len(x) \le 2.2 \cdot 2^{-j}.
\]
By \cref{lem:U-concentration} (for the same $j,t$), with probability $1-O(\alpha^2)$ we have
\[
|\mu^{(1)} - x| \le 0.05 \cdot 2^{-j}
\quad\text{and}\quad
|\mu^{(2)} - x| \le 0.05 \cdot 2^{-j}.
\]
Combining these, we have with probability $1-O(\al^2)$ that
\[
Z_{j,t}=1 \implies  |\mu^{(1)}|,\ |\mu^{(2)}|
    \le 2.25 \cdot 2^{-j},
\]
and therefore
\[
Z_{j,t}=1 \implies  |U_{j,t}|
    < 16 \cdot 4^{-j}.
\]
Thus, except with probability $1-O(\al^2)$, we have either $Z_{j, t} = 0$ or $|U_{j,t}| < 16 \cdot 4^{-j}$ (which implies that $U(j, t) = \Ut(j, t)$), and thus we are done.
\end{proof}

\begin{lemma}
For each $j\leq J$ and $t$, conditioned on $k_{j,t}$ we have
\[
\E[\rhat_{j,t}] = v_{k_{j,t}}^2 \cdot z_{j,t} + O(\alpha^2),
\]
where $z_{j,t} = \Pr[L(v_{k_{j,t}})=j \mid k_{j,t}]$ (as defined above).
\end{lemma}
\begin{proof}
Fix $j,t$ and condition on $k_{j,t}$ for the entire proof of this lemma. For brevity we write $U \coloneq  U_{j,t}$, $Z \coloneq  Z_{j,t}$ and $\widetilde U \coloneq  \Ut_{j,t}$. By \cref{lem:truncation},
\[
\Pr[ \widetilde U Z \neq U Z] = O(\alpha^2),
\]
and $|\widetilde U Z|,|U Z|\le 1$, so
\[
\bigl| \E[\widetilde U Z] - \E[U Z] \bigr|
    \le 2\Pr[ \widetilde U Z \neq U Z]
    = O(\alpha^2).
\]
Now, note that $UZ = \mu^{(1)} \mu^{(2)} Z$, and $\mu^{(1)},\mu^{(2)}$, $Z$ are independent given $k_{j,t}$, since they are computed from disjoint samples. Since $\E[\mu^{(1)}] = \E[\mu^{(2)}] = v_{k_{j,t}}$, and $\E[Z] = z_{j,t}$, the conclusion follows.
\end{proof}

\begin{corollary} \label{lem:expectation-qhat_j}
For each $j$, we have
\[
\E[\rhat_j] = \E_{k \sim [N]}[ v_k^2 \ind{L(v_k) = j} ] + O(\alpha^2).
\]
\end{corollary}
\begin{proof}
Recall $\rhat_j = \frac{1}{T_j} \sum_{t=1}^{T_j} \rhat_{j,t}$ and that each $k_{j,t}$ is independent and uniform in $[N]$. Thus, this follows directly from the previous lemma.
\end{proof}

Now, define $r_j = \E_{k \sim [N]}[v_k^2 \ind{L(v_k) = j}]$.

\begin{lemma}\label{lem:concentration-qhat_j}
For each $j \le J$, with probability $1 - O(\alpha^2)$ we have
\[
|\rhat_j - r_j| \le \Ot\big(\alpha \sqrt{r_j} + \alpha^2\big),
\]
where $r_j = \E_{k \sim [N]}[v_k^2 \ind{L(v_k)=j}]$.
\end{lemma}
\begin{proof}
Fix $j$.  Note that the random variables $\rhat_{j,1},\dots,\rhat_{j,T_j}$ are independent, nonnegative, and bounded:
\[
0 \le \rhat_{j,t} \le B_j \coloneqq 16 \cdot 4^{-j}.
\]
Recall that
\[
\rhat_j = \frac{1}{T_j} \sum_{t=1}^{T_j} \rhat_{j,t},
\]
and let $\mu_j = \E[\rhat_j]$. Note that $\mu_j = r_j + O(\alpha^2)$ by \cref{lem:expectation-qhat_j}.

By a Chernoff bound (\cref{lem:chernoff-01}) on $T_j \rhat_j = \sum \rhat_{j,t}$, with $B = B_j$ and $\gamma = \Theta(\log(1/\alpha))$, we have with probability at least $1 - O(\al^2)$ that
\[
|\rhat_j - \mu_j| \le \f{\gamma}{T_j}\bigl(
  \sqrt{B_j T_j \mu_j}
  + B_j
\bigr).
\]
We have $B_j / T_j = 16 \alpha^2$, so substituting this into the above expression gives
\[
|\rhat_j - \mu_j| \le \gamma(4\alpha \sqrt{\mu_j} + 16\alpha^2) = \Ot(\alpha \sqrt{\mu_j} + \alpha^2).
\]
By \cref{lem:expectation-qhat_j}, we have $\mu_j = r_j + O(\alpha^2)$, so
\[
\alpha \sqrt{\mu_j}
  = \alpha \sqrt{r_j + O(\alpha^2)}
  = \Ot\big(\alpha \sqrt{r_j} + \alpha^2\big),
\]
and
\[
|\mu_j - r_j| = O(\alpha^2).
\]
Combining these bounds, we get
\[
|\rhat_j - r_j|
  \le \Ot(\alpha \sqrt{r_j} + \alpha^2),
\]
except with probability $O(\alpha^2)$, as desired.
\end{proof}

Now, write
\[
r \coloneq  \E_{k \sim [N]}[v_k^2],.
\]
We wish to show that $|\qhat - \sqrt r| = \Ot(\al)$ with probability at least $2/3$.
Recall that we had earlier defined
\[
r_j = \E_{k \sim [N]}[v_k^2 \ind{L(v_k)=j}].
\]
Thus we have
\[
r = \sum_{j=0}^{J} r_j + r_{J+1}.
\]
By the first lemma, whenever $L(v_k)=J+1$ we have $|v_k| = O(\alpha)$ except with probability $O(\alpha^2)$ over the randomness defining $L$, so
\[
r_{J+1} = O(\alpha^2).
\]
By \Cref{lem:concentration-qhat_j} and a union bound over $j$, with probability $1-\Ot(\alpha^2)$ we have
\[
|\rhat_j - r_j| \le \Ot\big(\alpha \sqrt{r_j} + \alpha^2\big),
\quad\text{for all $0 \le j \le J$}.
\]
Assuming (for the rest of the proof) that this holds, we have
\begin{align*}
|\qhat^2 - r|
    &= \left|\sum_{j=0}^{J} \rhat_j - \sum_{j=0}^{J} r_j - r_{J+1} \right| 
    \le \sum_{j=0}^{J} \Ot\big(\alpha \sqrt{r_j} + \alpha^2\big) + O(\alpha^2) 
    = \Ot\big(\alpha \sqrt{r} + \alpha^2\big),
\end{align*}
using $\sum_j r_j \le r$ and $J = \Ot(1)$.

Finally, we bound the error in $\qhat$.  If $r \ge c\alpha^2$ for a suitable constant $c>0$, then
\[
|\qhat - \sqrt{r}|
    = \frac{|\qhat^2 - r|}{\qhat + \sqrt{r}}
    \le \frac{\Ot(\alpha \sqrt{r} + \alpha^2)}{\sqrt{r}}
    = \Ot(\alpha).
\]
If instead $r \le c\alpha^2$, then we also have $|\qhat^2| \le r + |\qhat^2 - r| = O(\alpha^2)$, so
\[
|\qhat - \sqrt{r}| \le \qhat + \sqrt{r} = O(\alpha).
\]
Thus in all cases
\[
|\qhat - \sqrt{r}| = \Ot(\alpha),
\]
and for small enough absolute constants in the algorithm the overall success probability is at least $2/3$, completing the proof of \Cref{thm:rademacher-norm} (noting that we can reduce $\al$ by a logarithmic factor to make the approximation error actually $\al$).


 \section*{Acknowledgements}
We thank Steve Flammia, John Wright, Allen Liu, and Ryan O'Donnell for helpful discussions and ChatGPT for general assistance.

\bibliography{references}
\bibliographystyle{alpha}
\appendix\newcommand{\osc}{\mathrm{osc}}
\section{Deferred Proofs}\label{sec:deferred}

\noisysearch*
\begin{proof}
Write $m_I:=\inf_{t\in I} g(t)$.
Since $g$ is continuous on the compact interval $I$, the infimum is attained.
We use the following elementary consequences of unimodality: 
Let $J=[u,u+w]\subseteq I$, and define
\[
q_1=u+\frac{w}{4},\qquad q_2=u+\frac{w}{2},\qquad q_3=u+\frac{3w}{4}.
\]
Then:
\begin{align*}
g(q_1)<g(q_2) &\implies \text{some minimizer of $g|_J$ lies in }[u,q_2],\\
g(q_3)<g(q_2) &\implies \text{some minimizer of $g|_J$ lies in }[q_2,u+w],\\
g(q_1)>g(q_2) &\implies \text{some minimizer of $g|_J$ lies in }[q_1,u+w],\\
g(q_3)>g(q_2) &\implies \text{some minimizer of $g|_J$ lies in }[u,q_3],\\
g(q_1)\ge g(q_2)\le g(q_3) &\implies \text{some minimizer of $g|_J$ lies in }[q_1,q_3].
\end{align*}
We now describe the algorithm.
Starting from $J_0:=I$, at stage $k$, we have a current interval
\[
J_k=[u_k,u_k+w_k]
\]
which, inductively, contains at least one global minimizer of $g$ on $I$.
If $\abs{J_k}\le \eta$, terminate and output any point of $J_k$.
Otherwise, define the quarter-points
\[
q_{k,1}=u_k+\frac{w_k}{4},\qquad q_{k,2}=u_k+\frac{w_k}{2},\qquad q_{k,3}=u_k+\frac{3w_k}{4}.
\]
Query the oracle at these three points.
We say the stage is \emph{decisive} if at least one of the following holds:
\begin{align}
\widetilde g(q_{k,1}) &\le \widetilde g(q_{k,2})-2\eta,\label{eq:dec1}\\
\widetilde g(q_{k,3}) &\le \widetilde g(q_{k,2})-2\eta,\label{eq:dec2}\\
\widetilde g(q_{k,1}) &\ge \widetilde g(q_{k,2})+2\eta,\label{eq:dec3}\\
\widetilde g(q_{k,3}) &\ge \widetilde g(q_{k,2})+2\eta.\label{eq:dec4}
\end{align}
If none of \eqref{eq:dec1}--\eqref{eq:dec4} holds, we stop and output the point among
\[
\{q_{k,1},q_{k,2},q_{k,3}\}
\]
with smallest $\widetilde g$-value.
If the stage is decisive, then by the oracle guarantee
$|\widetilde g(t)-g(t)|\le \eta$
we obtain the corresponding exact comparisons for $g$. For example,
\[
\widetilde g(q_{k,1})\le \widetilde g(q_{k,2})-2\eta
\implies
g(q_{k,1})\le g(q_{k,2}),
\]
and similarly for the other three comparisons.

We then update the current interval as follows:
\begin{itemize}
\item if \eqref{eq:dec1} holds, replace $J_k$ by $[u_k,q_{k,2}]$;
\item else if \eqref{eq:dec2} holds, replace $J_k$ by $[q_{k,2},u_k+w_k]$;
\item else if both \eqref{eq:dec3} and \eqref{eq:dec4} hold, replace $J_k$ by $[q_{k,1},q_{k,3}]$;
\item else if only \eqref{eq:dec3} holds, replace $J_k$ by $[q_{k,1},u_k+w_k]$;
\item else if only \eqref{eq:dec4} holds, replace $J_k$ by $[u_k,q_{k,3}]$.
\end{itemize}
The new interval still contains a minimizer of $g|_{J_k}$, hence in particular still contains a global minimizer of $g$ on $I$.
Moreover, every replacement interval has length at most $\frac{3}{4}|J_k|$.

It remains to justify that if none of these decisive steps occur, then we can stop. We also stop if $|J_k|\leq \eta$, since then the approximate optimality of any point in $J_k$ follows from $g$ being $10$-Lipschitz.
Suppose none of \eqref{eq:dec1}--\eqref{eq:dec4} holds at stage $k$.
Then
\[
|\widetilde g(q_{k,1})-\widetilde g(q_{k,2})|<2\eta,
\qquad
|\widetilde g(q_{k,3})-\widetilde g(q_{k,2})|<2\eta.
\]
Using the oracle error, we get
\[
|g(q_{k,1})-g(q_{k,2})|<4\eta,
\qquad
|g(q_{k,3})-g(q_{k,2})|<4\eta.
\]
Hence
$\operatorname{osc}_{\{q_{k,1},q_{k,2},q_{k,3}\}}(g)<8\eta$.
Since $g$ is $10$-regularized on $J_k$,
\[
\operatorname{osc}_{J_k}(g)
\le
10\cdot \operatorname{osc}_{\{q_{k,1},q_{k,2},q_{k,3}\}}(g)
<
80\eta.
\]
Because $J_k$ contains a global minimizer of $g$ on $I$, we have
$\inf_{t\in J_k} g(t)=m_I$.
Therefore every point of $J_k$ satisfies
\[
g(t)\le m_I+80\eta.
\]
In particular, if $\hat t$ is the sampled point among $\{q_{k,1},q_{k,2},q_{k,3}\}$ with smallest $\widetilde g$-value, then
\[
g(\hat t)\le m_I+80\eta.
\]
So the stopping rule is correct, with $C=80$.

Let $t_\star$ be any global minimizer of $g$ on $I$ that remains inside every current interval; such a minimizer exists by construction.
Because every update keeps $t_\star$ inside the current interval and shrinks the interval length by a factor at most $\frac{3}{4}$, after $N$ decisive stages we have an interval $J_N$ containing $t_\star$ with
\[
|J_N|\le \left(\frac{3}{4}\right)^N |I|.
\]
Thus the query complexity is controlled by how many multiplicative interval shrinks are needed to reach a region where the oscillation of $g$ is $O(\eta)$ around a minimizer. This takes $O(\log(1/\eta))$ decisive stages, since the original interval had width at most $10$.
Each stage uses $3$ oracle queries, so the total number of queries is also as above. 
Together with the stopping-rule guarantee, this proves the claim.
\end{proof}

We will use the following lemma to prove \cref{fact:g-regularized} and \cref{fact:f-regularized}.

\begin{lemma}\label{lem:cosine-shape}
Let $H(t)\coloneq \sqrt{A-B\cos(t-t_0)}$, where $A\ge B\ge 0$ and $B\le 1$.
Then:
\begin{enumerate}
    \item $H$ is $1$-Lipschitz on $\mathbb R$;
    \item $H$ is unimodal on every interval of length at most $\pi$;
    \item $H$ is $16$-regularized on every interval of length at most $\pi$.
\end{enumerate}
\end{lemma}

\begin{proof}
By translation invariance, we may assume $t_0=0$, so $H(t)=\sqrt{A-B\cos t}$.

\smallskip
\noindent\textbf{1. Lipschitz bound.}
Write $A-B\cos t=(A-B)+B(1-\cos t)=(A-B)+2B\sin^2(t/2)$.
At every point where $H(t)\neq 0$, the function $H$ is differentiable and
\[
H'(t)=\frac{B\sin t}{2H(t)}=\frac{B\sin(t/2)\cos(t/2)}{H(t)}.
\]
Since $A\ge B$, we have $H(t)\ge \sqrt{2B}\,|\sin(t/2)|$. Therefore
\[
|H'(t)| \le \frac{B|\sin(t/2)\cos(t/2)|}{\sqrt{2B}\,|\sin(t/2)|} = \sqrt{\frac B2}\,|\cos(t/2)| \le 1,
\]
using $B\le 1$. Thus $|H'|\le 1$ wherever $H$ is differentiable. Since $H$ is continuous on $\mathbb R$, it follows that $H$ is $1$-Lipschitz.

\smallskip
\noindent\textbf{2. Unimodality.}
If $B=0$, then $H$ is constant, hence unimodal. Assume $B>0$. Since $H$ is even and $2\pi$-periodic, and since $\cos t$ is strictly decreasing on $[0,\pi]$, the function $H$ is strictly increasing on $[0,\pi]$. Hence the only local extrema of $H$ occur at multiples of $\pi$, and any interval of length at most $\pi$ contains at most one such extremum in its interior. Therefore $H$ is unimodal on every interval of length at most $\pi$.

\smallskip
\noindent\textbf{3. Regularization.}
We first note a general fact: if $g\ge 0$ is $c$-regularized on an interval $I$, then $\sqrt g$ is $2c$-regularized on $I$.
Indeed, let $Q$ denote the set of quarter-points of $I$, and write $M_S\coloneq \max_S g$ and $m_S\coloneq \min_S g$ for $S\in\{I,Q\}$. If $M_I=0$, then $g\equiv 0$ on $I$ and the claim is trivial. Otherwise,
\[
\osc_I(\sqrt g)
= \sqrt{M_I}-\sqrt{m_I}
= \frac{M_I-m_I}{\sqrt{M_I}+\sqrt{m_I}}.
\]
Since $g$ is $c$-regularized, $M_I-m_I\le c(M_Q-m_Q)$, and hence
\[
\osc_I(\sqrt g)
\le c \cdot\frac{M_Q-m_Q}{\sqrt{M_I}+\sqrt{m_I}}
= c\cdot\osc_Q(\sqrt g)\cdot \frac{\sqrt{M_Q}+\sqrt{m_Q}}{\sqrt{M_I}+\sqrt{m_I}}.
\]
Because $Q\subseteq I$, we have $M_Q\le M_I$, so $\sqrt{M_Q}+\sqrt{m_Q}\le 2\sqrt{M_I}$, while $\sqrt{M_I}+\sqrt{m_I}\ge \sqrt{M_I}$. Therefore $\osc_I(\sqrt g)\le 2c \cdot \osc_Q(\sqrt g)$.

Now apply this to $g(t)\coloneq A-B\cos t$. If $B=0$, then $H$ is constant and there is nothing to prove. Assume $B>0$. For every set $S$, we have $\osc_S(g)=B\,\osc_S(\cos)$, so $g$ has exactly the same regularization constant as $\cos$. Thus it suffices to show that $\cos t$ is $8$-regularized on every interval of length at most $\pi$.

Let $I=[u,u+w]$ with $0<w\le \pi$, and let $q_1=u+w/4$, $q_2=u+w/2$, and $q_3=u+3w/4$.

\smallskip
\noindent\emph{Case 1: $I$ contains no multiple of $\pi$.}
Then $\cos t$ is monotone on $I$, so $\osc_I(\cos)=|\cos(u+w)-\cos u|$ and $\osc_{\{q_1,q_2,q_3\}}(\cos)=|\cos(u+3w/4)-\cos(u+w/4)|$. Using $\cos a-\cos b=2\sin((a+b)/2)\sin((b-a)/2)$, we get
\[
\frac{\osc_I(\cos)}{\osc_{\{q_1,q_2,q_3\}}(\cos)}
= \frac{2|\sin(u+w/2)\sin(w/2)|}{2|\sin(u+w/2)\sin(w/4)|}
= \frac{\sin(w/2)}{\sin(w/4)}
= 2\cos(w/4) \le 2.
\]

\smallskip
\noindent\emph{Case 2: $I$ contains a multiple of $\pi$.}
By reflection symmetry and invariance under shifting by $\pi$, we may assume that $I$ contains $0$, the maximum point of $\cos$. Write $I=[-\alpha,\beta]$, where $\alpha,\beta\ge 0$ and $\alpha+\beta\le \pi$, and by symmetry assume $\beta\ge \alpha$. Then $\osc_I(\cos)=1-\cos\beta$.

The quarter-points are $q_1=(\beta-3\alpha)/4$, $q_2=(\beta-\alpha)/2$, and $q_3=(3\beta-\alpha)/4$. Since $\beta\ge \alpha$, we have $q_3\ge \beta/2$. Also, one of $q_1,q_2$ has absolute value at most $\beta/4$: if $\alpha\ge \beta/2$, then $0\le q_2=(\beta-\alpha)/2\le \beta/4$, while if $\alpha<\beta/2$, then $-\beta/8<q_1=(\beta-3\alpha)/4\le \beta/4$. Call such a point $q_\ast$.

Since $\cos$ is even and decreasing on $[0,\pi]$, we have $\min_{t \in \{q_1,q_2,q_3\}}\cos(t) \le \cos(q_3)\le \cos(\beta/2)$ and $\max_{t \in \{q_1,q_2,q_3\}}\cos(t) \ge \cos(q_\ast)\ge \cos(\beta/4)$. Hence
\[
\osc_{\{q_1,q_2,q_3\}}(\cos)\ge \cos(\beta/4)-\cos(\beta/2).
\]
Let $x\coloneq \beta/4\in[0,\pi/4]$ and $c\coloneq \cos x\in[1/\sqrt2,1]$. Then
\[
1-\cos 4x = 2\sin^2 2x = 8\sin^2x\cos^2x = 8c^2(1-c^2),
\]
while
\[
\cos x-\cos 2x = c-(2c^2-1)=(1-c)(2c+1).
\]
Therefore
\[
\frac{\osc_I(\cos)}{\osc_{\{q_1,q_2,q_3\}}(\cos)} \le \frac{1-\cos 4x}{\cos x-\cos 2x}
= \frac{8c^2(1-c^2)}{(1-c)(2c+1)} = \frac{8c^2(1+c)}{2c+1}.
\]
Since $c\le 1$, we have $c^2(1+c)=c^2+c^3\le c+c=2c<2c+1$, so the last quantity is \(<8\). Thus $\cos$ is $8$-regularized on $I$.

Combining the two cases, $\cos$ is $8$-regularized on every interval of length at most $\pi$. Hence $g$ is $8$-regularized, and therefore $H=\sqrt g$ is $16$-regularized.
\end{proof}

\gregularized*
\begin{proof}
Fix $\theta\in[0,\pi/2]$ and write $g(\gamma)\coloneq g_\theta(\gamma)=\Delta(\theta,\gamma)$. From \eqref{eq:delta-squared-basic}, we have
\[
\frac12 g(\gamma)^2
= 1-|c_0|^2\cos^2\theta-|c_1|^2\sin^2\theta -2|c_0||c_1|\sin\theta\cos\theta\cos(\gamma-\gamma^\star),
\]
where $\gamma^\star=-\arg(c_0\overline{c_1})$. Equivalently, $g(\gamma)=\sqrt{A-B\cos(\gamma-\gamma^\star)}$, with
\[
A\coloneq 2-2|c_0|^2\cos^2\theta-2|c_1|^2\sin^2\theta,
\qquad B\coloneq 4|c_0||c_1|\sin\theta\cos\theta.
\]
Moreover,
\[
A-B = 2-2\bigl(|c_0|\cos\theta+|c_1|\sin\theta\bigr)^2 \ge 0,
\]
and
\[
B=(2|c_0||c_1|)(2\sin\theta\cos\theta)\le 1,
\]
since $2|c_0||c_1|\le |c_0|^2+|c_1|^2\le 1$ and $2\sin\theta\cos\theta\le 1$. Thus \Cref{lem:cosine-shape} applies. Therefore, $g_\theta$ is unimodal on every interval of length at most $\pi$, is $1$-Lipschitz, and is $16$-regularized.
\end{proof}

\fregularized*
\begin{proof}
From \eqref{eq:f-squared-second}, we have $\frac12 f(\theta)^2 = 1-s^2\cos^2(\theta-\theta^\star)$. Using $\cos^2 x=\frac12(1+\cos 2x)$, this becomes
\[
f(\theta)^2 = 2-s^2-s^2\cos\bigl(2(\theta-\theta^\star)\bigr).
\]
Define $H(t)\coloneq \sqrt{2-s^2-s^2\cos(t-2\theta^\star)}$. Then $f(\theta)=H(2\theta)$. Here $A\coloneq 2-s^2$ and $B\coloneq s^2$, so $A\ge B\ge 0$ and $B\le 1$. Thus \Cref{lem:cosine-shape} applies to $H$.

Therefore $H$ is unimodal and $16$-regularized on every interval of length at most $\pi$, and is $1$-Lipschitz. Let $I=[u,u+w]\subseteq[0,\pi/2]$, and let $q_1=u+w/4$, $q_2=u+w/2$, and $q_3=u+3w/4$. Then $2I=[2u,2u+2w]$ has length at most $\pi$, and its quarter-points are precisely $2q_1,2q_2,2q_3$. Hence
\[
\osc_I(f)
= \osc_{2I}(H) \le 16\cdot\osc_{\{2q_1,2q_2,2q_3\}}(H)= 16\cdot\osc_{\{q_1,q_2,q_3\}}(f).
\]
So $f$ is $16$-regularized on $[0,\pi/2]$.

Since $H$ is unimodal on every interval of length at most $\pi$, the function $f(\theta)=H(2\theta)$ is unimodal on $[0,\pi/2]$. Finally,
\[
|f(\theta)-f(\theta')|
= |H(2\theta)-H(2\theta')| \le |2\theta-2\theta'| = 2|\theta-\theta'|,
\]
so $f$ is $2$-Lipschitz.
\end{proof}

\end{document}